\documentclass{article}

\usepackage{hhline}
\usepackage[dvipsnames,table,xcdraw]{xcolor}

\usepackage{latexsym,amsfonts,amsmath,amsthm,amssymb}
\usepackage{hyperref}
\usepackage{relsize}
\usepackage[htt]{hyphenat}

\usepackage[latin1]{inputenc}
\usepackage{epsfig, amssymb, amsmath}
\usepackage{latexsym, graphics, graphicx}
\usepackage{url}
\usepackage{upquote}

\usepackage{epsfig}
\usepackage{float}
\usepackage[normalem]{ulem} 
\usepackage{url}

\usepackage{algpseudocode, algorithm, float}

\long\def\comment#1{}

\newcommand{\caC}{\ensuremath{\mathcal C}}

\newcommand{\caR}{\ensuremath{\mathcal R}}

\newcommand{\caT}{\ensuremath{\mathcal T}}

\newcommand{\caX}{\ensuremath{\mathcal X}}

\newcommand{\pair}[1]{\ensuremath{\langle}{#1}\ensuremath{\rangle}}

\newcommand{\sort}[1]{\ensuremath{\mathsf{#1}}}

\newcommand{\Variables}{\caX}
\newcommand{\Symbols}{\Sigma}

\newcommand{\TermsOn}[5]{{\caT^{#4}_{#1}(#2)}_{#3}^{#5}}
\newcommand{\Terms}{\TermsOn{\Symbols}{\Variables}{}{}{}}
\newcommand{\TermsS}[1]{\TermsOn{\Symbols}{\Variables}{\sort{#1}}{}{}}

\newcommand{\GTermsOn}[2]{\caT^{#2}_{#1}}
\newcommand{\GTerms}{\GTermsOn{\Symbols}{}}
\newcommand{\GTermsS}[1]{\GTermsOn{\Symbols,\sort{#1}}{}}

\newcommand{\SubstOn}[2]{{\cal S}ubst(#1,#2)}
\newcommand{\Substs}{\SubstOn{\Symbols}{\Variables}{}{}{}}
\newcommand{\idsubst}{\textit{id}}

\newcommand{\composeSubst}{}
\newcommand{\composeRel}{;}
\newcommand{\compose}{\composeSubst}
\newcommand{\restrict}[2]{#1|_{#2}}

\newcommand{\congr}[1]{=_{#1}}

\newcommand{\csu}[3]{\textit{CSU}_{#3}({#1})}
\newcommand{\csuV}[3]{\textit{CSU\/}^{#2}_{#3}({#1})}

\newcommand{\var}[1]{\mathit{Var}(#1)}
\newcommand{\occ}[1]{\mathit{Pos}(#1)}
\newcommand{\occSub}[2]{\mathit{Pos}_{#2}(#1)}
\newcommand{\funocc}[1]{\mathit{Pos}_{\Symbols}(#1)}

\newcommand{\subterm}[2]{#1|_{#2}}
\newcommand{\replace}[3]{#1[#3]_{#2}}
\newcommand{\domain}[1]{\mathit{Dom}(#1)}
\newcommand{\range}[1]{\intrvar{#1}}
\newcommand{\intrvar}[1]{\mathit{Ran}(#1)}

\newcommand{\rootpos}{\mathsmaller{\Lambda}}

\newcommand{\unif}[1][]{\mathop{=}}

\newcommand{\rewrite}[1]{\rightarrow_{#1}}
\newcommand{\rewrites}[1]{\rightarrow^*_{#1}}
\newcommand{\Grewrites}[1]{\mathop{\stackrel{?}{\rightarrow}\!\!{}^*_{#1}}}

\newcommand{\norm}[1]{{\downarrow_ {#1}}}

\newcommand{\sem}[1]{{[\![#1]\!]}_{E,B}}

\usepackage{ifthen}

\newenvironment{flemma-noname}[2][]{\vskip\topsep\noindent{\bf
Lemma #2\ifthenelse{\equal{#1}{}}{}{\ }#1.}\em\ }{\vskip\topsep}

\newcommand{\ignore}[1]{}

\newtheorem{thm}{Theorem}
\newtheorem{exa}{Example}
\newtheorem{lem}[thm]{Lemma}
\newtheorem{definition}{Definition}

\providecommand{\keywords}[1]
{
  \small	
  \textbf{\textit{Keywords---}} #1
}

\title{An Efficient Canonical Narrowing Implementation with
Irreducibility and SMT Constraints for Generic Symbolic Protocol Analysis}

\author{
    Ra{\'u}l L{\'o}pez-Rueda \footnote{rloprue@upv.es}
    \and
    Santiago Escobar \footnote{sescobar@upv.es}
    \and
    Julia Sapi\~{n}a \footnote{jsapina@upv.es}
}

\date{VRAIN, Universitat Polit\`ecnica de Val\`encia, Valencia, Spain, Camino de vera, S/N, 46022 Valencia Spain}

\begin{document}

\maketitle

\begin{abstract}
Narrowing and unification are very useful tools for symbolic analysis of rewrite theories, and thus for any model that can be specified in that way. A very clear example of their application is the field of formal cryptographic protocol analysis, which is why narrowing and unification are used in tools such as Maude-NPA, Tamarin and Akiss. In this work we present the implementation of a canonical narrowing algorithm, which improves the standard narrowing algorithm, extended to be able to process rewrite theories with conditional rules. The conditions of the rules will contain SMT constraints, which will be carried throughout the execution of the algorithm to determine if the solutions have associated satisfiable or unsatisfiable constraints, and in the latter case, discard them.
\end{abstract}


%%%%%%%%%%%%%%%%%%%%%%%%%%%%%%%%%%%%%%%%%%%%%%%%%%%%%%%%%%%%%%%%%%%%%%%%%%%%%%%%%%%%%%%%
%%%%%%%%%%%%%%%%%%%%%%%%%%%%%%%%%%%%%%%%%%%%%%%%%%%%%%%%%%%%%%%%%%%%%%%%%%%%%%%%%%%%%%%%
\section{Introduction}
\label{sec:introduction}
Verification of protocol security properties modulo the algebraic properties of a protocol's cryptographic functions for an arbitrary number of sessions is generally undecidable, and the state space is infinite. \emph{Symbolic} techniques such as unification and narrowing modulo a protocol's algebraic properties, as well as SMT solving, are particularly well suited to support symbolic model checking and theorem proving verification methods. 

The Maude-NPA \cite{DBLP:conf/fosad/EscobarMM07} is a symbolic model checker for cryptographic protocol analysis based on the above-mentioned symbolic techniques, which are efficiently supported by the underlying Maude language \cite{DBLP:journals/jlap/DuranEEMMRT20}. These Maude-based symbolic techniques are also used by other protocol analysis tools such as Tamarin \cite{DBLP:conf/cav/MeierSCB13} and AKISS \cite{DBLP:journals/tocl/ChadhaCCK16}.

State explosion is a significant challenge in this kind of symbolic model checking analysis modulo algebraic properties, particularly because unification modulo algebraic properties can generate large numbers of unifiers when computing symbolic transitions. Although Maude-NPA has quite effective state space reduction techniques \cite{DBLP:journals/iandc/EscobarMMS14}, further state space reduction gains can be obtained by more sophisticated equational narrowing techniques such as \emph{canonical narrowing}~\cite{DBLP:conf/birthday/EscobarM19}.
One of the contributions of the present
paper is an implementation of canonical narrowing, an algorithm that can achieve further computational and performance improvements. For this we use the new unification and narrowing features supported by the current Maude 3.3, as well as its meta-level features. In addition, the features of Maude also allow us to create an extension of the algorithm in which calls can be made to an SMT solver (e.g., Yices2~\cite{10.1007/978-3-319-08867-9_49}) to check the satisfiability of \emph{SMT constraints}. These restrictions will allow us to express new models on which to apply narrowing.
Throughout this work, we consider several experimental examples in order to demonstrate the effectiveness of the new implementation in Maude 3.3.

A first example of a module in which a great impact is noticed when using canonical narrowing versus standard narrowing consists of 
an equational theory that includes the idempotence property, allowing us to eliminate the redundant elements in a term. The reason why we have chosen this property is because it is highly problematic\footnote{Maude provides an idempotence operator attribute \texttt{idem}, but it cannot be used together with some other operator attributes, because the appropriate matching and normalization algorithms have not been developed yet, see \cite[\S 20.3.1]{maude-manual}.} in automated reasoning. It makes the representation of sets easier, in contrast to multi-sets, and it is useful when dealing with processes or agents: if we have several identical processes working at the same time, the idempotence property allows us to eliminate one of them to avoid redundancy and reduce the use of computational resources.

\begin{exa}\label{exa:idem-vending}
We can specify a vending machine, in which dollars and quarters (of sort \texttt{Money}) are inserted to buy combinations of coffee and apples (of sort \texttt{Item}). To do this, we 
specify that each coffee costs one dollar and each apple three-quarters of a dollar. Two rules handle state transitions for those specifications. 

We add some equations that express idempotence, i.e., two dollars can be replaced by a dollar, two apples can be replaced by one apple, and two coffees can be replaced by one coffee. 
This makes no sense in real life, we acknowledge it, but it is not intended to be realistic (see Section~\ref{sec:experiments} for more realistic examples), just an example to illustrate the impact of idempotence on performance, allowing us to compare standard and canonical narrowing. 

Note the addition of a variable \texttt{M} of type \texttt{Marking} to make the rules and equations strict ACU-coherent, see \cite{DBLP:journals/tcs/Meseguer17}. Also note that the rules are coherent with respect to the equations, see \cite{DBLP:journals/jlap/Meseguer20}.

\clearpage
{\scriptsize
\begin{verbatim}
mod IDEMPOTENCE-VENDING-MACHINE is
   sorts Coin Item Marking Money State . 
   subsort Coin < Money .   subsorts Money Item < Marking .

   op empty : -> Money .
   op __ : Money Money -> Money [assoc comm id: empty] .
   op __ : Marking Marking -> Marking [assoc comm id: empty] .
   op <_> : Marking -> State . 
   ops $ q : -> Coin . ops c a : -> Item .
   var M : Marking .

   rl [buy-c] : < M $ > => < M c > [narrowing] .    
   rl [buy-a] : < M $ > => < M a q > [narrowing] .
   
   eq [idem-dollar] : $ $ M = $ M [variant] .   
   eq [idem-item-a] : a a M = a M [variant] .
   eq [change] : q q q q M = $ M [variant] .    
   eq [idem-item-c] : c c M = c M [variant] .
endm
\end{verbatim}
}

\noindent
Note that idempotence is not specified for quarters (\texttt{q}) because there is already an equation that reduces the repetition of four quarters to a dollar, so that adding idempotence for quarters would 
create a problem of confluence.

If we consider an initial term \texttt{< M1 >} that only contains a variable of type \texttt{Money}, we would obtain several traces by using the narrowing algorithm. In each one of the observed narrowing states, it is necessary to unify with the left-hand side of the rules to determine the new narrowing steps that can be taken. Each of those possible steps creates a new branch in the search tree. One of these traces takes us to the term \texttt{< \$ a c q q M4 >}, which also contains a variable \texttt{M4} of type \texttt{Money}. The narrowing sequence associated to this term is as follows (rule and/or equation renaming may be needed to avoid variable clashes):
$$\texttt{<\;M1\;>} \: \leadsto_{\texttt{buy-a},\sigma_1} \: \texttt{<\;\$\;a\;q\;M2\;>} \: \leadsto_{\texttt{buy-c},\sigma_2} \: \texttt{<\;a\;c\;q\;M3\;>} \: \leadsto_{\texttt{buy-a},\sigma_3} \: \texttt{<\;\$\;a\;c\;q\;q\;M4\;>}$$
where \texttt{M2} and \texttt{M3} are also variables of type \texttt{Money}
and the computed substitutions are
$\sigma_1=\{\texttt{M1}\mapsto\texttt{\$\;M2}\}$,
$\sigma_2=\{\texttt{M2}\mapsto\texttt{\$\;M3}\}$, 
and
$\sigma_3=\{\texttt{M3}\mapsto\texttt{\$\;M4}\}$. 
Note that in the first narrowing step, 
the substitution applied to the left-hand side of the rule
\texttt{buy-a} is 
$\rho_1=\{\texttt{W1}\mapsto\texttt{\$\;M2}\}$.
For the second narrowing step, 
the substitution applied to the left-hand side of the rule
\texttt{buy-c} is 
$\rho_2=\{\texttt{W2}\mapsto\texttt{a\;q\;M3}\}$.
For the third narrowing step, 
the substitution applied to the left-hand side of the rule
\texttt{buy-a} is 
$\rho_3=\{\texttt{W3}\mapsto\texttt{\$\;a\;c\;q\;M4}\}$.
Note that extra \texttt{\$} are introduced by $\rho_1$ and $\rho_3$ due to equational unification using
the variant equations and the axioms.
\end{exa}

As we will see later, the use of canonical narrowing will allow us to introduce \emph{irreducibility constraints} in the algorithm, which in many cases will significantly reduce the number of branches in the narrowing search tree.

It is not hard to notice that if canonical narrowing succeeds in reducing the number of states in the search tree in many cases (reducing computation time), it will most likely have a positive impact on performance when applying narrowing to the symbolic analysis of protocols. This type of protocol analysis allows us to determine whether an attacker can cause a protocol to fail any of its security objectives. But
in many protocols, it is necessary to represent distances, time, or coordinates using real numbers. The formal analysis of this type of protocols can be done using either an explicit model with physical information, or by using an abstract model without physical information, e.g., untimed, and showing it is sound and complete with respect to a model with physical information. The former is more intuitive for the user, but the latter is often chosen because not all cryptographic protocol analysis tools support reasoning about, e.g., time or space. 
SMT solvers allow precisely the use of explicit models with physical information, translating the physical laws into SMT constraints. In order to analyze these models using narrowing algorithms, there is a need to extend them so that they are capable of handling these restrictions. One way to do it is by having narrowing to handle conditional rules, as in \cite{DBLP:journals/jlap/Meseguer20}, in which each of the constraints will be collected at runtime.
In the time and space sequences below, we show one of the protocols that uses physical laws to define time and space (and therefore, be able to function correctly as it is specified).
This protocol
goes beyond existing narrowing approaches such as \cite{DBLP:conf/birthday/EscobarM19,DBLP:journals/jlap/Meseguer20,DBLP:conf/indocrypt/Aparicio-Sanchez20,DBLP:conf/birthday/Aparicio-Sanchez21}, since two cryptographic primitives are combined, exclusive-or over a set of nonces and a commitment scheme, apart of time and location represented as real numbers, requiring both irreducibility and SMT constraints.

\begin{exa}\label{exa:brands-and-chaum}
The Brands-Chaum protocol \cite{BC93} specifies communication  between a verifier V and a prover P.  P needs to authenticate itself to V, and also needs to prove that it is within a distance ``d" of it. We use a couple of messages below, $N_V$ and $N_P \oplus N_V$, as the rapid exchanged messages proving to be within the distance ``d".

Below, we describe a typical interaction between the prover and 
the verifier. We shall use the following notation: $N_A$ denotes a nonce generated by $A$, $S_A$ denotes a secret generated by $A$,
$X {;} Y$ denotes concatenation of two messages $X$ and $Y$,
$\textit{commit}(N,S)$ denotes commitment of  secret $S$ with a nonce $N$,  
$\textit{open}(N,S,C)$ denotes opening a commitment $C$ using the nonce $N$ and checking whether it carries the secret $S$,
$\oplus$ is the exclusive-or operator, and $\textit{sign}(A,M)$ denotes  $A$ signing message $M$. 

\noindent
{\small%
\begin{align*}
P &\rightarrow V :\textit{commit}(N_P,S_P) \\[-1ex]
& \mbox{//The prover sends his name and a commitment}\\
V &\rightarrow P :N_V \\[-1ex]
& \mbox{//The verifier sends a nonce %}\\[-1ex]
%& \mbox{//
and records the time when this message was sent}\\[-1ex]
P &\rightarrow V : N_P \oplus N_V \\[-1ex]
& \mbox{//The verifier checks %the answer
whether the %}\\[-1ex]
%& \mbox{//
message arrives within two times a fixed distance}\\
P &\rightarrow V : S_P \\[-1ex]
& \mbox{//The prover sends the committed secret %}\\[-1ex]
%& \mbox{//
and the verifier
opens the commitment
}\\
P &\rightarrow V : \mathit{sign}_P(N_V ; N_P \oplus N_V) \\[-1ex]
& \mbox{//The prover signs the two rapid exchange messages}
\end{align*}
}
\end{exa}

We assume
the participants are located at an arbitrary given topology 
(participants do not move from their assigned locations)
with distance constraints, 
where travelled time and coordinates are represented by real  numbers.

We assumed coordinates $P_x$, $P_y$, $P_z$ for each participant $P$.

The previous informal Alice\&Bob notation was naturally extended to include time in \cite{DBLP:conf/indocrypt/Aparicio-Sanchez20}
and to include both time and location
in
\cite{DBLP:conf/birthday/Aparicio-Sanchez21}.
First, we add the time when a message was sent or received as a subindex $P_{t_1} \to V_{t_2}$.
Second, 
 
the sending and receiving times of a message differ by the distance between them just by adding some location constraints
$$\lfloor d(A,B) \rfloor :=\  (d(A,B)\geq 0 \wedge d(A,B)^2=(A_x-B_x)^2+(A_y-B_y)^2+(A_z-B_z)^2)$$

Third, the distance bounding constraint of the verifier is represented as an arbitrary distance $d$.
Time and space constraints 
are written using 
quantifier-free formulas
in 
real arithmetic.
For convenience, 
we 
allow both
$2*x = x + x$ and the monus function
$x \dot{-} y =
\textit{if}\ y < x\ \textit{then}\ x - y\ \textit{else}\ 0$
as definitional
extensions.

\begin{exa}\label{exa:brands-and-chaum-2}
(Continued Example~\ref{exa:brands-and-chaum})
In the time and space interaction sequences below, a vertical bar differentiates between the process interactions and the corresponding constraints 
associated to the metric space. 
The following action sequence from \cite{DBLP:conf/birthday/Aparicio-Sanchez21} differs from \cite{DBLP:conf/indocrypt/Aparicio-Sanchez20} only on the terms $\lfloor d(P,V)\rfloor$.

\noindent
{\small
\[
\begin{array}{@{}r@{}l@{\;}l@{}}
P_{t_1} \rightarrow V_{t'_1} &\::\: \textit{commit}(N_P,S_P) 
& \mid t'_1 = t_1 + d(P,V) \wedge \lfloor d(P,V)\rfloor\\
V_{t_2} \rightarrow P_{t'_2} &\::\: N_V & \mid t'_2 = t_2 + d(P,V) \wedge t_2 \geq t'_1 \wedge \lfloor d(P,V)\rfloor\\\
P_{t_3} \rightarrow V_{t'_3} &\::\: N_P \oplus N_V & \mid t'_3 = t_3 + d(P,V) \wedge t_3 \geq t'_2 \wedge \lfloor d(P,V)\rfloor\\\
V & \::\: t'_3\: \dot{-}\: t_2 \leq 2*d\\
P_{t_4} \rightarrow V_{t'_4} &\::\: S_P & \mid t'_4 = t_4 + d(P,V)\wedge t_4 \geq t_3 \wedge \lfloor d(P,V)\rfloor\\\
V & \::\: \textit{open}(N_P,S_P,\textit{commit}(N_P,S_P))\\
P_{t_5} \rightarrow V_{t'_5} &\::\: \mathit{sign}_P(N_V ; N_P \oplus N_V) & \mid t'_5 = t_5 + d(P,V) \wedge t_5 \geq t_4 \wedge \lfloor d(P,V)\rfloor\
\end{array}%
\]%
}

\begin{figure}[t]
\begin{minipage}{.5\linewidth}
\centering
\includegraphics[width=.7\linewidth]{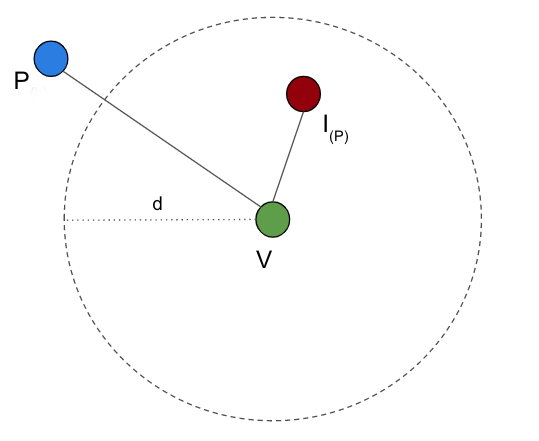}
\caption{Mafia Attack}
\label{fig:mafia}
\end{minipage}
\begin{minipage}{.5\linewidth}
\centering
\includegraphics[width=.6\linewidth]{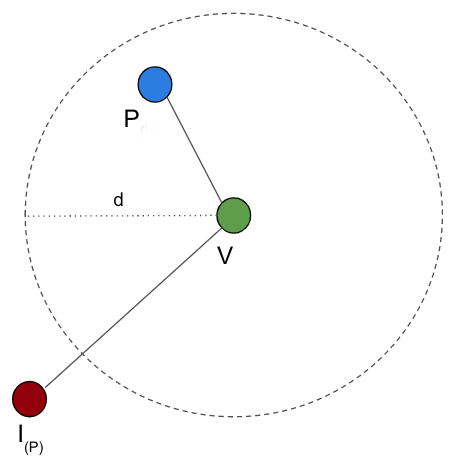}
\caption{Hijacking Attack}
\label{fig:hijacking}
\end{minipage}
\end{figure}

\noindent
The Brands-Chaum protocol is designed to defend against mafia frauds, where an honest prover is outside the neighborhood of the verifier 
(i.e., $d(P,V) > d$)
but an intruder is inside (i.e., $d(I,V) \leq d$), pretending to be the honest prover
as depicted in Figure~\ref{fig:mafia}.
The following is an example of an \emph{attempted} mafia fraud,
in which 
the intruder simply forwards messages back and forth between the prover and the verifier.
We write $I(P)$ to denote an intruder pretending to be an honest prover $P$.\\

\noindent
{\small
\[
\begin{array}{@{}r@{}r@{}l@{}l@{}l@{}}
P_{ t_1} &\rightarrow &I_{t_2} &: \textit{commit}(N_P,S_P) 
& \mid t_2 = t_1 + d(P,I)  \wedge \lfloor d(P,I)\rfloor\\
I(P)_{t_2} &\rightarrow &V_{t_3}&:  \textit{commit}(N_P,S_P) 
& \mid t_3 = t_2 + d(V,I) \wedge \lfloor d(V,I)\rfloor\\
V_{t_3} &\rightarrow &I(P)_{t_4} &: N_V & \mid t_4 = t_3 + d(V,I) \wedge \lfloor d(V,I)\rfloor\\
I_{t_4} &\rightarrow &P_{t_5} &: N_V & \mid t_5 = t_4 + d(P,I) \wedge \lfloor d(P,I)\rfloor\\
P_{t_5} &\rightarrow &I_{t_{6}} &: N_P \oplus N_V & \mid t_{6} = t_5 + d(P,I) \wedge \lfloor d(P,I)\rfloor\\
I(P)_{t_{6}} &\rightarrow &V_{t_{7}} &: N_P \oplus N_V & \mid t_{7} = t_{6} + d(V,I) \wedge \lfloor d(V,I)\rfloor\\
& & V & : t_{7} \dot{-} t_3 \leq 2*d\\
P_{t_{8}} &\rightarrow &I_{t_{9}} &: S_P & \mid t_{9} = t_{8} + d(P,I) \wedge t_8 \geq t_5 \wedge \lfloor d(P,I)\rfloor\\
I(P)_{t_{10}} &\rightarrow &V_{t_{11}} &: S_P & \mid t_{11} = t_{10} + d(V,I)  \wedge t_{11} \geq t_7 \wedge \lfloor d(V,I)\rfloor\\
I(P)_{t_{12}} &\rightarrow &V_{t_{13}} &: sign_P(N_V ; N_P \oplus N_V) & \mid t_{13} = t_{12} + d(V,I)  \wedge t_{13} \geq t_{11}  \wedge \lfloor d(V,I)\rfloor
\end{array}
\]%
}

\noindent
This attack is physically unfeasible, since 
it would require that $2 * d(V,I) + 2* d(P,I) \leq 2*d$, which is unsatisfiable
by $d(V,P) > d > 0$ and the triangular inequality $d(V,P) \leq d(V,I) + d(P,I)$, satisfied in  three-dimensional space.
This attack was already proved unfeasible in \cite{DBLP:conf/indocrypt/Aparicio-Sanchez20} using only  metric space assumptions
and
in \cite{DBLP:conf/birthday/Aparicio-Sanchez21} using a Euclidean space. 

However, a distance hijacking 
attack is possible (i.e., the time and distance constraints are satisfiable),
as depicted in Figure~\ref{fig:hijacking}. 
In this scenario, an intruder located outside the neighborhood of the verifier  (i.e., $d(V,I) > d$) 
succeeds in convincing the verifier that he is inside the neighborhood by 
exploiting the presence of an honest prover in the neighborhood (i.e., $d(V,P) \leq d$) to achieve his goal.
The following is an example of a \emph{successful} distance hijacking, 
in which 
the intruder listens to the exchanged messages between the prover and the verifier but 
sends the last message.

\noindent
{\small
\[
\begin{array}{@{}r@{}l@{\;}l@{}}
P_{t_1} \rightarrow V_{t_2}~~~~~ &\::\: \textit{commit}(N_P,S_P) 
& \mid t_2 = t_1 + d(P,V) \wedge \lfloor d(P,V)\rfloor\\
V_{t_2} \rightarrow P_{t_3},I_{t'_3} &\::\: N_V & \mid t_3 = t_2 + d(P,V)  \wedge \lfloor d(P,V)\rfloor \\
&& \mid t'_3 = t_2 + d(I,V)  \wedge \lfloor d(V,I)\rfloor\\
P_{t_3} \rightarrow V_{t_4},I_{t'_4} &\::\: N_P \oplus N_V & \mid t_4 = t_3 + d(P,V)  \wedge \lfloor d(P,V)\rfloor \\
&& \mid t'_4 = t_3 + d(I,P) \wedge \lfloor d(I,P)\rfloor\\
V~~~~~~~ & \::\: t_4\: \dot{-}\: t_2 \leq 2*d\\
P_{t_5} \rightarrow V_{t_6}~~~~~ &\::\: S_P & \mid t_6 = t_5 + d(P,V) \wedge \lfloor d(P,V)\rfloor\\
&& \mid t_5 \geq t_3 \wedge t_6 \geq t_4 \\
I(P)_{t_7} \rightarrow V_{t_8}~~~~~ &\::\: \mathit{sign}_I(N_V ; N_P \oplus N_V) & \mid t_8 = t_7 + d(I,V) \wedge \lfloor d(I,V)\rfloor\\
&& \mid t_7 \geq t'_4 \wedge t_8 \geq t_6  \\
\end{array}
\]%
}%

\noindent
This attack was proved feasible in \cite{DBLP:conf/indocrypt/Aparicio-Sanchez20} using  metric space assumptions, and it was also possible in  three-dimensional space
in \cite{DBLP:conf/birthday/Aparicio-Sanchez21}.
\end{exa}

Note that to model the above protocol in Maude, it is necessary to use narrowing, in addition to SMT constraints. This is exactly what is proposed in this paper.

\paragraph{Plan of the paper}
The rest of this work is organized as follows. Section \ref{sec:preliminaries} provides some preliminaries on rewriting logic and narrowing. Section \ref{sec:canonical-narrowing} gives a detailed presentation of our new canonical narrowing with irreducibility and SMT constraints. Section \ref{sec:implementation} describes our new implementation of 
that algorithm in Maude 3.3. Section \ref{sec:experiments} presents the experiments, among which can be found the use of standard and canonical narrowing algorithms, both with unconditional modules and conditional modules using SMT constraints.

Finally, Section~\ref{sec:conclusions} summarizes the paper and presents some future work.

\paragraph{Contributions}
Let us clarify the contributions of this paper in detail. This work is an extended version of \cite{DBLP:conf/wrla/Lopez-Rueda22,DBLP:conf/wrla/Lopez-Rueda22-invited}. 
\cite{DBLP:conf/wrla/Lopez-Rueda22} presents a new implementation of the canonical narrowing strategy of \cite{DBLP:conf/birthday/EscobarM19}, which was already an improvement of the contextual narrowing of
\cite{DBLP:conf/esorics/ErbaturEKLLMMNSS12}.
\cite{DBLP:conf/wrla/Lopez-Rueda22-invited} extends the canonical narrowing strategy of \cite{DBLP:conf/wrla/Lopez-Rueda22}
to handle SMT constraints as part of conditional rewrite theories.
Proofs of all the results of \cite{DBLP:conf/wrla/Lopez-Rueda22-invited} are included here.
We have combined the implementations of \cite{DBLP:conf/wrla/Lopez-Rueda22,DBLP:conf/wrla/Lopez-Rueda22-invited}, as well as made changes and improvements to it.
Both standard and canonical narrowing, as well as their extended versions to process SMT constraints of conditional modules in Maude, can be invoked.
We have re-executed all the experiments that appeared in \cite{DBLP:conf/wrla/Lopez-Rueda22}
in order to demonstrate the greater performance. We have also added a new example (Section \ref{subsec:bank-account}) with their respective experiments. This example has never been run with narrowing before, since Maude's built-in narrowing is not capable of processing conditional rewrite theories natively.
The experiments of \cite{DBLP:conf/wrla/Lopez-Rueda22-invited} have been redefined, so that now in Sections \ref{subsec:brands-chaum-time} and \ref{subsec:brands-chaum-time-space} we 
do not just look for protocol vulnerabilities, but we also do so by comparing the performance of standard narrowing and canonical narrowing. These two examples 
are important for this work, since 
they combine (i) SMT constraints on non-linear real arithmetic, 
(ii) the exclusive-or theory,
and
(iii) the commitment theory between the participants. 

In \cite{DBLP:conf/icfem/Lopez-RuedaE22}, we have developed a narrowing algorithm that allows the use of conditional rules with variant-based equalities in contrast to this work that allows conditional rules with only SMT conditions. Their combination is left as future work.

%%%%%%%%%%%%%%%%%%%%%%%%%%%%%%%%%%%%%%%%%%%%%%%%%%%%%%%%%%%%%%%%%%%%%%%%%%%%%%%%%%%%%%%%
%%%%%%%%%%%%%%%%%%%%%%%%%%%%%%%%%%%%%%%%%%%%%%%%%%%%%%%%%%%%%%%%%%%%%%%%%%%%%%%%%%%%%%%%
\section{Preliminaries}
\label{sec:preliminaries}
We follow the classical notation and terminology from \cite{Terese03} for term rewriting, and from \cite{Meseguer92,DBLP:journals/jlap/Meseguer20} for rewriting logic and order-sorted notions.

We assume an order-sorted signature $\Sigma$ 
with a poset of sorts $(S, \leq)$. The poset $(\sort{S},\leq)$ of sorts for $\Symbols$
is partitioned into equivalence classes, called \emph{connected components}, by the equivalence relation $(\leq \cup \geq)^+$. We assume that each connected component $[\sort{s}]$  has a \emph{top element} under $\leq$, denoted $\top_{[\sort{s}]}$ and called the \emph{top sort} of $[\sort{s}]$. This involves no real loss of generality, since if $[\sort{s}]$ lacks a top sort, it can be easily added.

We  assume an $\sort{S}$-sorted family $\Variables=\{\Variables_\sort{s}\}_{\sort{s} \in \sort{S}}$ of disjoint variable sets with each $\Variables_\sort{s}$ countably infinite. $\TermsS{s}$ is the set of terms of sort \sort{s}, and $\GTermsS{s}$ is the set of ground terms of sort \sort{s}. We write $\Terms$ and $\GTerms$ for the corresponding order-sorted term algebras. Given a term $t$, $\var{t}$ denotes the set of variables in $t$. The notation $x{:}\sort{s}$ indicates the sort \sort{s} of a variable $x$.

Positions
are represented by sequences of natural numbers denoting an access path in
the term when viewed as a tree.  The top or root position is denoted by the empty sequence $\rootpos$.
We define the relation $p \leq q$ between positions as
$p \leq p $ for any $p$;
and
$p \leq p.q$ for any $p$ and $q$.
Given $U
\subseteq \Symbols\cup\Variables$, $\occSub{t}{U}$ denotes the set of positions of
a term $t$ that are rooted by symbols or variables in $U$.
The set of positions of a term $t$ is written $\occ{t}$,
and
the set of non-variable positions $\funocc{t}$.
The subterm of $t$
at position $p$
is $\subterm{t}{p}$ and $\replace{t}{p}{u}$ is
the term $t$ where $\subterm{t}{p}$ is
replaced by $u$.

A \textit{substitution} $\sigma\in\Substs$ is a sorted mapping from a finite subset of $\Variables$ to $\Terms$. Substitutions are written as $\sigma=\{X_1 \mapsto t_1,\ldots,X_n \mapsto t_n\}$ where the domain of $\sigma$ is $\domain{\sigma}=\{X_1,\ldots,X_n\}$ and the set of variables introduced by terms $t_1,\ldots,t_n$ is written $\range{\sigma}$. The identity substitution is $\idsubst$. Substitutions are homomorphically extended to $\Terms$. The application of substitution $\sigma$ to a term $t$ is denoted by $t\sigma$ or $\sigma(t)$.
The restriction of substitution $\sigma$ to variables $W$ is denoted $\subterm{\sigma}{W}=\{(X \mapsto t)\in\sigma \mid X \in W\}$.

A \textit{$\Symbols$-equation} is an unoriented pair $t = t'$, where $t,t' \in \TermsS{s}$ for some sort $\sort{s}\in\sort{S}$.   Given
$\Symbols$ and a set $E$ of $\Symbols$-equations, order-sorted equational logic induces a congruence relation $\congr{E}$ on terms $t,t' \in \Terms$ (see~\cite{Meseguer97}). Throughout this paper we assume that $\GTermsS{s}\neq\emptyset$ for every sort \sort{s}, because this affords a simpler deduction system (see \cite{DBLP:journals/sigplan/GoguenM81}). We write ${\TermsOn{\Symbols/E}{\Variables}{}{}{}}$ and ${\GTermsOn{\Symbols/E}{}}$ for the corresponding order-sorted term algebras modulo the congruence closure $\congr{E}$, denoting the equivalence class of a term $t\in\Terms$ as $[t]_E \in {\TermsOn{\Symbols/E}{\Variables}{}{}{}}$.

The first-order language of equational $\Symbols$-formulas is defined as: $\Symbols$-equations $t = t'$ as basic atoms, conjunction $\wedge$ of formulas,
disjunction $\vee$ of formulas,
negation $\neg$ of a formula,
universal quantification $\forall$ of a variable $x{:}\sort{s}$ in a formula,
and 
existential quantification $\exists$ of a variable $x{:}\sort{s}$ in a formula.
A formula is quantifier-free (QF) if it does not contain any quantifier.
Given a $\Symbols$-algebra $A$, a formula $\varphi$, and an assignment $\alpha\in X\mapsto A$ for the free variables $X$ in $\varphi$, 
$A,\alpha \models \varphi$ denotes that 
$\varphi\alpha$ is satisfied\footnote{We do not consider specific satisfiability techniques and refer the reader to \cite{10.1007/978-3-319-45641-6_26}.} and $A \models \varphi$ holds if $\forall\alpha : A,\alpha\models\varphi$.

An \emph{equational theory} $(\Symbols,E)$ is a pair with $\Symbols$ an order-sorted signature and $E$ a set of $\Symbols$-equations. An equational theory $(\Symbols,E)$ is \emph{regular} if for each $t = t'$ in $E$, we have $\var{t} = \var{t'}$. An equational theory $(\Symbols,E)$ is \emph{linear} if for each $t = t'$ in $E$, each variable occurs only once in $t$ and in $t'$. An equational theory $(\Symbols,E)$ 
is \textit{sort-preserving} if for each $t = t'$ in $E$, each sort \sort{s}, and each substitution $\sigma$, we have $t \sigma \in \TermsS{s}$ iff $t' \sigma \in \TermsS{s}$. An equational theory $(\Symbols,E)$ is \emph{defined using top sorts} if for each equation $t = t'$ in $E$, all variables in $\var{t}$ and $\var{t'}$ have a top sort.
Given two equational theories $G=(\Symbols,E)$ and $T=(\Symbols_0,\Gamma)$, we say $T$ is the background theory of 
$G$ iff $\Symbols_0 \subseteq \Symbols$
and for each ground $\Symbols_0$-formula $\varphi$, 
$\GTermsOn{\Symbols/E}{} \models \varphi \iff 
T \models \varphi$.

An \textit{$E$-unifier} for a $\Symbols$-equation $t = t'$ is a substitution $\sigma$ such that $t\sigma \congr{E} t'\sigma$.  For $\var{t}\cup\var{t'} \subseteq W$, a set of substitutions $\csuV{t = t'}{W}{E}$ is said to be a \textit{complete} set of unifiers for the 
$\Symbols$-equation $t = t'$ modulo $E$ away from $W$ iff: (i) each $\sigma \in \csuV{t = t'}{W}{E}$ is an $E$-unifier of $t = t'$; (ii) for any $E$-unifier $\rho$ of $t = t'$ there is a substitution $\sigma \in \csuV{t=t'}{W}{E}$ such that $\subterm{\sigma}{W} \sqsupseteq_{E} \subterm{\rho}{W}$ (i.e., there is a substitution $\eta$ such that $\subterm{(\sigma\compose\eta)}{W} \congr{E} \subterm{\rho}{W}$); and (iii) for all $\sigma \in \csuV{t=t'}{W}{E}$, $\domain{\sigma} \subseteq (\var{t}\cup\var{t'})$ and $\range{\sigma} \cap W = \emptyset$.

A \textit{conditional rewrite rule} is an oriented pair  $l \to r \mbox{ if }\varphi$, where $l \not\in \Variables$,
$\varphi$ is a QF $\Symbols_0$-formula, and $l,r \in \TermsS{s}$ for some sort $\sort{s}\in\sort{S}$. 
An unconditional rewrite rule is written $l \to r$.
A \textit{conditional order-sorted rewrite theory} is a 
tuple $(\Symbols,E,R,T)$ with $\Symbols$ an order-sorted signature, $E$ a set of $\Symbols$-equations,
$T$ is the background theory 
$(\Symbols_0, \Gamma)$
of 
$(\Symbols, E)$, and $R$ a set of 
(conditional) rewrite rules.
The set $R$ of rules is \textit{sort-decreasing} if for each $l \rightarrow r$ (or $l \rightarrow r \mbox{ if }\varphi$) in $R$, each $\sort{s} \in \sort{S}$, and each substitution $\sigma$, $r\sigma \in \TermsS{s}$ implies $l\sigma \in \TermsS{s}$. 

The rewriting relation on $\Terms$, written $t \rewrite{R} t'$  or $t \rewrite{p,R} t'$ holds between $t$ and $t'$ iff there exist a position $p \in \funocc{t}$, a rule $l \to r\mbox{ if }\varphi\in R$ and a substitution $\sigma$, such that $T \models \varphi\sigma$, $\subterm{t}{p} = l\sigma$, and $t' = \replace{t}{p}{r\sigma}$. The relation $\rewrite{R/E}$ on $\Terms$ is ${\congr{E} \composeRel\rewrite{R}\composeRel\congr{E}}$. The transitive (resp. transitive and reflexive) closure of $\rewrite{R/E}$ is denoted $\rewrite{R/E}^+$ (resp. $\rewrites{R/E}$).  A term $t$ is called $\rewrite{R/E}$-irreducible (or just $R/E$-irreducible) if there is no term $t'$ such that $t \rewrite{R/E} t'$.  For $\rewrite{R/E}$ confluent and terminating, the irreducible version of a term $t$ is denoted by $t\norm{R/E}$.
A substitution $\sigma$ is in
irreducible (or normalized) form
if
$\forall x: \sigma(x)=\sigma(x)\norm{R/E}$.

A relation $\rewrite{R,E}$ on $\Terms$ is defined as: $t \rewrite{p,R,E} t'$ (or just $t \rewrite{R,E} t'$) iff there exist a  position $p \in \funocc{t}$, a rule $l \to r\mbox{ if }\varphi$ in $R$, and a substitution $\sigma$ such that $T \models \varphi\sigma$, $\subterm{t}{p} \congr{E} l\sigma$ and $t' = \replace{t}{p}{r\sigma}$. Reducibility of $\rewrite{R/E}$ is undecidable in general since $E$-congruence classes can be arbitrarily large. Therefore, $R/E$-rewriting is usually implemen\-ted~\cite{JouannaudK86} by $R,E$-rewriting under some conditions on $R$ and $E$ such as confluence, termination, and coherence.

We call $(\Symbols,B,\vec{E})$  a \emph{decomposition} of an order-sorted equational theory 
\linebreak $(\Symbols,E\cup B)$ if $B$ is regular, linear, sort-preserving, defined using top sorts, and has a finitary and complete unification algorithm,  which implies that $B$-matching is decidable, and the equations $E$ oriented into rewrite rules $\vec{E}$ are \emph{convergent}, i.e., confluent, terminating, and strictly coherent \cite{DBLP:journals/tcs/Meseguer17} modulo $B$, and sort-decreasing. 

Given a decomposition $(\Symbols,B,E)$ of an equational theory, $(t',\theta)$ is an $E,B$-\emph{variant}~\cite{comon-delaune,Escobar-JLAP} (or just a variant) of term $t$ if $t\theta\norm{\vec{E},B} \congr{B} t'$ and $\theta\norm{\vec{E},B} \congr{B} \theta$. 
Given two variants
$(t',\sigma)$ and $(t'',\theta)$ 
of a term $t$,
we say 
$(t'',\theta)$ is more general than
$(t',\sigma)$,
written
$(t'',\theta) \sqsupseteq_{E,B} (t',\sigma)$, 
if there is a substitution $\rho$ such that $t' \congr{B} t''\rho$ and $\restrict{\sigma}{\var{t}} =_{B} \restrict{(\theta\rho)}{\var{t}}$.
A \emph{complete set of $E,B$-variants}~\cite{Escobar-JLAP} (up to renaming) of a term $t$ is a subset, denoted by $\sem{t}$, of the set of all $E,B$-variants of $t$ such that, for each $E,B$-variant $(t',\sigma)$ of $t$, there is an $E,B$-variant $(t'', \theta) \in \sem{t}$ such that 
$(t'',\theta) \sqsupseteq_{E,B} (t',\sigma)$.
A decomposition $(\Symbols,B,E)$  has the \emph{finite variant property} (FVP)~\cite{Escobar-JLAP} (also called a \emph{finite variant decomposition}) iff for each $\Symbols$-term $t$, a complete set $\sem{t}$ of its most general variants is finite.

\begin{definition}[Reachability goal]
    Given an order-sorted rewrite theory \linebreak $(\Symbols,G,R,T)$, a \emph{reachability goal} is defined as a pair $t \Grewrites{R/G} t'$, where $t, t' \in \TermsS{s}$. It is abbreviated as $t \Grewrites{} t'$ when the theory is  clear from the context; $t$ is the \emph{source} of the goal and $t'$ is the \emph{target}. A substitution $\sigma$ is a $R/G$-\emph{solution} of the reachability goal (or just a solution for short) iff 
either $\sigma(t) =_G \sigma(t')$ or    
    there is a sequence $\sigma(t) \rewrite{R/G} u_1 \rewrite{R/G} \cdots \rewrite{R/G} u_{k-1} \rewrite{R/G} \sigma(t')$.

    A set $\Gamma$ of substitutions is said to be a \emph{complete set
    of solutions} of $t \Grewrites{R/G} t'$ iff (i) every substitution $\sigma \in \Gamma$ is a solution of $t \Grewrites{R/G} t'$, and (ii) for any solution $\rho$ of $t \Grewrites{R/G} t'$, there is a substitution $\sigma \in \Gamma$  more general than $\rho$ modulo $G$,   i.e., $\subterm{\sigma}{\var{t}\cup\var{t'}} \sqsupseteq_{G} \restrict{\rho}{\var{t}\cup\var{t'}}$.
\end{definition}
\noindent

This provides a tool-independent semantic framework for 
analysis of rewrite theories, including, e.g., protocols considering algebraic properties. Note that we have
not added to the preliminaries the condition $\var{\varphi}\cup\var{r} \subseteq \var{l}$ for rewrite rules $l\to r \mbox{ if }\varphi\in R$ and thus 
extra variables may appear in righthand sides or conditions of rules which should be properly instantiated by a solution substitution $\sigma$.

If the terms $t$ and $t'$ in a goal $t \Grewrites{R/G} t'$ are ground and rules have no extra variables in their right-hand sides, then  goal solving becomes a standard rewriting reachability problem. However, since we allow terms $t,t'$ with variables, we need a mechanism  more general than standard rewriting to find solutions of reachability goals. \emph{Narrowing} with $R$ modulo $G$ generalizes rewriting by performing {\em unification\/}  at non-variable positions instead of the usual matching modulo $G$. 

Soundness and completeness of narrowing for solving reachability goals are proved in \cite{JouannaudK86,meseguer-thati-hosc06} 
for unconditional rules $R$ modulo an equational theory $G$ 
and
in 
\cite{DBLP:journals/jlap/Meseguer20} 
for conditional rules $R$ modulo an equational theory $G$,
both with the restriction of considering only
order-sorted \emph{topmost} rewrite theories, i.e., rewrite theories where all the rewrite steps happen at the top of the term.

%%%%%%%%%%%%%%%%%%%%%%%%%%%%%%%%%%%%%%%%%%%%%%%%%%%%%%%%%%%%%%%%%%%%%%%%%%%%%%%%%%%%%%%%
%%%%%%%%%%%%%%%%%%%%%%%%%%%%%%%%%%%%%%%%%%%%%%%%%%%%%%%%%%%%%%%%%%%%%%%%%%%%%%%%%%%%%%%%
\section{Canonical Narrowing with Irreducibility and SMT Constraints}
\label{sec:canonical-narrowing}

When $(\Sigma,E\cup B)$ has a decomposition as $(\Sigma,B,\vec{E})$, then the initial algebra $\caT_{\Sigma/E\cup B}$
is isomorphic to the canonical term algebra $\caC_{\Sigma/E\cup B}{=}(C_{\Sigma/E\cup B},{\rightarrow_{R/E \cup B}})$,
where $C_{\Sigma/E\cup B}=\{C_{\Sigma/E\cup B,\sort{s}}\}_{\sort{s}\in\sort{S}}$
and
$C_{\Sigma/E\cup B,\sort{s}}=\{[t\norm{\vec{E},B}]_{B} \in T_{\Sigma/B}\mid t\norm{\vec{E},B} \in T_{\Sigma,\sort{s}}\}$
and where for each $f\in\Sigma$,
$f_{\caC_{\Sigma/E\cup B}}([t_1]_B,\ldots,[t_n]_B)=$
	
$[f(t_1,\ldots,t_n)\norm{\vec{E},B}]_B$.

We have an isomorphism of initial algebras
$\caT_{\Sigma/E\cup B} \cong \caC_{\Sigma/E\cup B}$. 
Likewise, we have an isomorphism of free $(\Sigma,E\cup B)$-algebras 
$\caT_{\Sigma/E\cup B}(\caX) \cong \caC_{\Sigma/E\cup B}(\caX)$,
where $\caC_{\Sigma/E\cup B}(\caX)=(C_{\Sigma/E\cup B}(\caX),\rightarrow_{R/E\cup B})$
and
$$C_{\Sigma/E\cup B,\sort{s}}(\caX)=\{[t\norm{\vec{E},B}]_{B} \in T_{\Sigma/B}(\caX)\mid t\norm{\vec{E},B} \in T_{\Sigma}(\caX)_{\sort{s}}\}.$$
The key point of canonical rewriting is that we can simulate rewritings 
\linebreak
$[t]_{E\cup B} \rightarrow_{R/E\cup B} [t']_{E\cup B}$
by corresponding rewritings 
$[t\norm{\vec{E},B}]_{B} \rightarrow_{R/E,B} [t'\norm{\vec{E},B}]_{B}$
and make rewriting decidable when $(\Sigma,B,E)$ is FVP.
Since completeness of narrowing is satisfied for topmost rewrite theories (see \cite{meseguer-thati-hosc06}), 
we assume a sort \texttt{State} such that all rewriting are at the top of terms of that sort.

\begin{definition}[SMT Canonical Rewriting]
Let $\caR=(\Sigma,E\cup B,R,T)$ be a topmost order-sorted rewrite theory such that $(\Sigma,E\cup B)$
has an FVP decomposition $(\Sigma,B,\vec{E})$.
Let 
$\caC^\circ_{\Sigma/E\cup B}(\caX)_{\sort{State}}=\{t\norm{\vec{E},B}\mid t\norm{\vec{E},B}\in T_{\Sigma}(\caX)_{\sort{State}}\}$,
so that $\caC^{\circ}_{\Sigma/E\cup B}(\caX)_{\sort{State}} \subseteq T_{\Sigma}(\caX)_{\sort{State}}$.
We then define the $\rightarrow_{R/E,B}$ canonical rewrite relation with rules $R$ modulo $E\cup B$ as the following
binary relation 
$\rightarrow_{R/E,B} \subseteq \caC^\circ_{\Sigma/E\cup B}(\caX)_{\sort{State}} \times \caC^\circ_{\Sigma/E\cup B}(\caX)_{\sort{State}}$, where
$t \rightarrow_{R/E,B} t'$
iff 
$\exists l\to r \mbox{ if }\varphi\in R$
and
$\exists \theta$ with $\domain{\theta} \subseteq \var{l}\cup\var{r}\cup\var{\varphi}$ and
$\theta=\theta\norm{\vec{E},B}$
such that:
(i)~$T\models\varphi\theta$,
(ii)~$(l\theta)\norm{\vec{E},B} \congr{B} t$,
and
(iii)~$t'\congr{B} (r\theta)\norm{\vec{E},B}$.
\end{definition}

The claim that 
$\rightarrow_{R/E,B}$ exactly captures/bisimulates the $\rightarrow_{R/E \cup B}$ rewrite relation is justified by the following result,
adapted from \cite[Thm. 2]{DBLP:conf/birthday/EscobarM19}.

\begin{thm}
For each $t,t'\in T_\Sigma(\caX)_\sort{State}$,
$t \rightarrow_{R/E\cup B} t'$ 
iff $t\norm{\vec{E},B} \rightarrow_{R/E,B} t'\norm{\vec{E},B}$.
\end{thm}
\begin{proof}
Since $\rightarrow_{R/E,B} \subseteq \rightarrow_{R/E\cup B}$, we only need to prove the ($\Rightarrow$) direction.
If $t \rightarrow_{R/E\cup B} t'$, we have a substitution $\theta$ and a rule $l\to r  \mbox{ if }\varphi\in R$ such that
$t \congr{E\cup B} l\theta$,
$T\models\varphi\theta$,
 and $t' \congr{E\cup B} r\theta$.
But then, of course, we also have 
$t \congr{E\cup B} l(\theta\norm{\vec{E},B})$ 
and 
$t' \congr{E\cup B} r(\theta\norm{\vec{E},B})$. 
By the Church-Rosser Theorem modulo $B$ \cite[Thm. 3]{var-sat-scp}, we then have
$$
\begin{array}{lllllll}
t & \congr{E\cup B} & l(\theta\norm{\vec{E},B}) & \rightarrow_{R} & r(\theta\norm{\vec{E},B}) & \congr{E\cup B} & t' \\
\norm{\vec{E},B} & & \norm{\vec{E},B} & & \norm{\vec{E},B}  & & \norm{\vec{E},B} \\
t\norm{\vec{E},B} & \congr{B} & l(\theta\norm{\vec{E},B})\norm{\vec{E},B} & \rightarrow_{R/E,B} & r(\theta\norm{\vec{E},B})\norm{\vec{E},B} & \congr{B} & t'\norm{\vec{E},B} \\
\end{array}
$$
as desired.
\qed
\end{proof}

A term $t(x_1{:}s_1,\ldots,x_n{:}s_n)$ 
can be viewed as a symbolic, effective method to describe a (typically infinite) set of terms, namely the set
$$\lceil t(x_1{:}s_1,\ldots,x_n{:}s_n)\rceil = \{t(u_1,\ldots,u_n)\mid u_i\in\TermsS{s_i}\}=\{t\theta \mid \theta\in\Substs\}.$$
We think as $t$ as a \emph{pattern}, which symbolically describes all its \emph{instances} (including 
also the non-ground ones).
Since $(\Sigma,B,E)$
is a decomposition of an equational theory $(\Sigma,E\cup B)$, 
we can consider only irreducible instances of $t$
$$\lceil t\rceil_{\vec{E},B} = \{(t\theta)\norm{\vec{E},B} \mid \theta\in\Substs\}$$

However, since we are interested in terms that may satisfy some irreducibility and SMT constraints, we can obtain a more expressive symbolic pattern language
where patterns are \emph{constrained by both irreducibility and SMT constraints}. 
That is, we consider constrained patterns of the form
$\pair{t,\Pi,\varphi}$
where $\Pi$ is a finite set of irreducible terms
and $\varphi$ is a QF $\Symbols$-formula.
Then we can define:

\noindent
\begin{align*}
\lceil \pair{t,(u_1,\ldots,u_k),\varphi} \rceil_{\vec{E},B} =  \{(t\theta)\norm{\vec{E},B}\ \mid&
\ \theta\in\Substs,
T \models \varphi\theta,\\ 
&\ u_1\theta,\ldots,u_k\theta \mbox{ are $E,B$-irreducible} \}.
\end{align*}

The canonical narrowing relation $\leadsto_{R/E,B}$ includes irreducibility constraints only for the left-hand sides of the rules and SMT constraints only from the conditional part of the rules.

\begin{definition}[SMT Canonical  Narrowing]\label{def:canonical-narrowing}
    Given a topmost order-sorted rewrite theory $(\Symbols,E\cup B,R,T)$ such that $(\Sigma,B,\vec{E})$ is a decomposition of $(\Sigma,E\cup B)$, the \emph{canonical narrowing relation with irreducibility and SMT constraints}  holds between $\pair{t,\Pi,\varphi}$ and $\pair{t',\Pi',\varphi'}$, denoted  
    $$\pair{t,\Pi,\varphi} \leadsto_{\alpha,R/E,B} \pair{t',\Pi',\varphi'}$$ iff there exists $l\to r\mbox{ if }\gamma\in R$, which we always assume renamed, so that $\var{\pair{t,\Pi,\varphi}}\cap(\var{r}\cup\var{l}\cup\var{\gamma})=\emptyset$, and a unifier $\alpha\in\csuV{t = l}{W}{E\cup B}$, where $W=\var{\pair{t,\Pi,\varphi}}\cup\var{r}\cup\var{l}\cup\var{\gamma}$, and
    \begin{enumerate}
        \item $\pair{t',\Pi',\varphi'} = \pair{r\alpha, \Pi\alpha \cup \{(l\alpha)\norm{\vec{E},B}\}, \varphi\alpha \wedge \gamma\alpha}$,
        \item $\Pi\alpha \cup \{(l\alpha)\norm{\vec{E},B}\}$ are $E,B$-irreducible, and
        \item $\varphi'$ is satisfiable, i.e., $\exists\alpha'$ s.t. $T \models \varphi'\alpha'$.
    \end{enumerate}
\end{definition}

\noindent

Note that we do not require a narrowing step to compute $\csu{t \unif l}{}{E\cup B}$ anymore, we perform regular equational unification but impose an irreducibility constraint on the  normal form of the instantiated left-hand side, which can be handled in Maude by using asymmetric unification~\cite{DBLP:conf/cade/ErbaturEKLLMMNSS13}, i.e., 
equational unification is done with irreducibility constraints~\cite[\S 14.10]{maude-manual}.

Irreducibility constraints are computed by using the irreducible left-hand side of the rules that are used in the narrowing steps. 
SMT constraints are simply added to the second component and checked for satisfiability.
Note that we assume that satisfiability of 
QF $\Symbols$-formulas is decidable, indeed for a subsignature $\Symbols_0 \subseteq \Symbols$ associated to the background theory $T$. Maude is using (see \cite[\S 16]{maude-manual}) the Yices2 SMT solver \cite{10.1007/978-3-319-08867-9_49} for satisfiability. Additionally, instructions on how to compile Maude with other SMT solvers are available in the compilation guide that comes with the Maude sources\footnote{Available at \url{https://github.com/SRI-CSL/Maude}}. For example, Maude-SE\footnote{Available at \url{https://maude-se.github.io/}.} also supports the use of the Z3 SMT Solver \cite{demoura2008z}.

In this way, each branch of the search tree will carry irreducibility constraints and SMT constraints (if any). 
In each new narrowing step, the list of irreducibility constraints computed previously in that branch must be taken into account, so that if it is necessary to reduce one of the terms appearing in the list to compute a new step, it will be discarded. 
Similarly, the SMT formula carried along the branch must be taken into account, so that if it becomes unsatisfiable after one narrowing step, it will be discarded. 

Therefore, we eliminate redundancy as well as branches of the search tree, which will be less and less wide than the tree resulting from using standard narrowing. In some cases, we will even get infinite search trees to become finite, ensuring termination.

The key completeness property about this relation is the following.

\begin{lem}[Lifting Lemma]
Given $\pair{t,\Pi,\varphi}$, a $E,B$-irreducible substitution $\theta$,
and terms $u,v\in\caC^\circ_{\Sigma/E,B}(\Variables)$ such that
$u=(t\theta)\norm{\vec{E},B}$,
$T \models \varphi\theta$, and $\Pi\theta$ are $E,B$-irreducible
and $u\to_{R/E,B} v$,
there is a canonical narrowing step with irreducibility and SMT constraints
$$\pair{t,\Pi,\varphi} \leadsto_{\alpha,R/E,B} \pair{r\alpha,\Pi\alpha \cup \{(l\alpha)\norm{\vec{E},B}\},\varphi'}$$
and a $E,B$-irreducible substitution $\gamma$ such that the following diagram holds (note that at the bottom part we consider ground instances due to $\lceil\cdot\rceil_{\vec{E},B}$)
$$
\begin{array}{lcl}
~~\pair{t,\Pi,\varphi} &\leadsto_{\alpha,R/E,B} &\pair{r\alpha,\Pi\alpha\cup\{(l\alpha)\norm{\vec{E},B}\},\varphi'}\\
~~~\downarrow_{\theta} & & \downarrow_{\gamma}\\
\lceil \pair{t,\Pi,\varphi}\rceil_{\vec{E},B} & \to_{R/E,B} &\lceil\pair{r\alpha,\Pi\alpha\cup\{(l\alpha)\norm{\vec{E},B}\},\varphi'}\rceil_{\vec{E},B}
\end{array}
$$
\begin{enumerate}
\item[(i)] $\theta =_B \restrict{(\alpha\gamma)}{\var{\pair{t,\Pi,\varphi}}}$,
\item[(ii)] 
$(r\alpha\gamma)\norm{\vec{E},B} =_B v$,
\item[(iii)] $\Pi\alpha\gamma \cup \{((l\alpha)\norm{\vec{E},B})\gamma\}$ are $E,B$-irreducible,
\item[(iv)] $T \models \varphi'\gamma$.
\end{enumerate}
\end{lem}
\begin{proof}
The rewriting step $u \to_{R/E,B} v$ exactly means that there are 
$l\to r \mbox{ if }\gamma\in R$ and
a $E,B$-irreducible substitution $\beta$ with 
$\domain{\beta}=\var{l}\cup\var{r}\cup\var{\gamma}$ 
such that 
$u =_B (l\beta)\norm{\vec{E},B}$,
$T\models\gamma\beta$, and $v=_B(r\beta)\norm{\vec{E},B}$.
But, since $\var{\pair{t,\Pi,\varphi}}\cap(\var{l}\cup\var{r}\cup\var{\gamma})=\emptyset$
by the rule renaming assumption, this exactly means that
$\theta\uplus \beta$ is a $E\cup B$-unifier of the unification problem $t = l$.
Therefore, there is an 
$E\cup B$-unifier $\alpha\in CSU_{E\cup B}^W (t = l)$
with
$W=\var{\pair{t,\Pi,\varphi}}\cup\var{r}\cup\var{l}\cup\var{\gamma}$
and a $E,B$-irreducible substitution $\rho$
such that
$$\theta\uplus\beta =_B \restrict{(\alpha\rho)}{W}$$
which gives (i).
But 
then
$$((t\alpha)\norm{\vec{E},B} = (l\alpha)\norm{\vec{E},B},\alpha) 
\sqsupseteq_{E,B}
((t\theta)\norm{\vec{E},B} = (l\beta)\norm{\vec{E},B}, \restrict{(\theta\uplus\beta)}{\var{t = l}})
$$
and this means that 
$(\Pi\alpha\cup\{(l\alpha)\norm{\vec{E},B}\}\gamma) =_B \Pi\theta\cup\{(l\alpha)\norm{\vec{E},B}\gamma\} =_B \Pi\theta\cup\{(l\beta)\norm{\vec{E},B}\}$ which proves (iii)
by the assumption that $\Pi\theta$ is $E,B$-irreducible.
Also, $\gamma$ is indeed satisfiable with $\beta$ which gives (iv).
Finally we also get $(r\alpha\gamma)\norm{\vec{E},B} =_B (r\beta)\norm{\vec{E},B} =_B v$,
proving (ii).
\qed
\end{proof}

Note that this shows that $v\in\lceil \pair{r\alpha,\Pi\alpha \cup \{(l\alpha)\norm{\vec{E},B}\},\varphi'}\rceil_{\vec{E},B}$.
Soundness is trivially proved.
\begin{lem}[Soundness]
Given 
an SMT canonical narrowing step 
$$\pair{t,\Pi,\varphi} \leadsto_{\alpha,R/E,B} \pair{r\alpha,\Pi\alpha \cup \{(l\alpha)\norm{\vec{E},B}\},\varphi'}$$
and a $E,B$-irreducible substitution $\gamma$ 
such that
$T \models \varphi'\gamma$
and
$\Pi\alpha\gamma\allowbreak \cup\allowbreak \{(l\alpha\gamma)\norm{\vec{E},B}\}$ are $E,B$-irreducible,
there is
an SMT canonical rewriting step 
$t\alpha\gamma \to_{R/E,B} (r\alpha\gamma)\norm{\vec{E},B}$.
\end{lem}

To further illustrate the operation of irreducibility constraints, an example is presented below in which a case of unification without irreducibility restrictions is shown, and another with the use of these, to see the differences in the solutions.

\begin{exa}\label{exa:idem-vending-variants}
If we look at the module of Example \ref{exa:idem-vending}, we can define an equational unification problem of the form $t \unif t'$. 
Specifically, if we consider the narrowing trace shown in that example
$$\texttt{< M1 >} \: \leadsto_{\sigma_1} \: \texttt{< \$\;a\;q\;M2 >} \: \leadsto_{\sigma_2} \: \texttt{< a\;c\;q\;M3 >} \: \leadsto_{\sigma_3} \: \texttt{< \$\;a\;c\;q\;q\;M4 >}$$
we can place ourselves in the third term, just before taking the last step. To compute the next possible steps from that term, it is necessary to try to unify it with the left-hand side of each of the defined rules. In this case, we will focus on the rule \texttt{buy-a}, which is also used to take the first step of the trace. The specification of the unification problem would then be $t=\texttt{< a\;c\;q\;M3 >}$ and $t'=\texttt{< W3 \$ >}$, where \texttt{W3} is a variable of type \texttt{Marking} (money, items, or combinations of them)
corresponding to the variable of a renamed version of rule \texttt{buy-a}. If we run the unification problem 
without any irreducible constraint 
using Maude's command, we will get $5$ unifiers as a solution:

{\scriptsize
\begin{verbatim}
Maude> variant unify < a c q M3:Money > =? < W3:Marking $ > .

            Unifier #1                                  Unifier #2                        
            M3:Money --> $ %1:Money                     M3:Money --> q q q #1:Money       
            W3:Marking --> q c a %1:Money               W3:Marking --> c a #1:Money       

            Unifier #3                                  Unifier #4                         
            M3:Money --> $ #1:Money                     M3:Money --> $ q q q %1:Money      
            W3:Marking --> $ q c a #1:Money             W3:Marking --> c a %1:Money        

                                Unifier #5
                                M3:Money --> q q q %1:Money
                                W3:Marking --> $ c a %1:Money

\end{verbatim}
}

\noindent
Note that $\rho_3$ of Example~\ref{exa:idem-vending} corresponds to the third unifier. But of those 5 unifiers, there are 3 that could be ignored, since the accumulated substitution makes the left-hand side of the \texttt{buy-a} rule used at the first narrowing step reducible.
Canonical narrowing would have computed irreducibility constraints that come from normalizing the instantiated left-hand side of the rules when taking the first and second step. 
That is, the terms
\texttt{<~M3\;\$~>}
(i.e., 
$\texttt{< W1 \$ >}\rho_1\norm{\vec{E},B}=\texttt{< \$\;\$\;M2 >}\norm{\vec{E},B}=\texttt{< \$\;M2 >}$
and 
$\texttt{<~\$\;M2~>}\sigma_2\norm{\vec{E},B}=\texttt{< M3\;\$ >}$) and \texttt{<~a\;q\;M3~>} 
(i.e., 
$\texttt{< W2 \$ >}\rho_2\norm{\vec{E},B}=\texttt{< a\;q\;M3 >}$)
are assumed to be irreducible when we want to take the last step of the trace. 
Maude's unification command allows us to indicate this irreducibility constraint using ``\texttt{such that M3 \$ irreducible}" right after the command (see \cite[\S 14.10]{maude-manual}). If we run it now, we can see how the number of unifiers found is reduced to 2, 
since the first, third and fourth unifiers from the previous command are discarded:

{\scriptsize
\begin{verbatim}
Maude> variant unify < a c q M3:Money > =? < W3:Marking $ > 
>      such that M3:Money $ irreducible .
            Unifier #1                                  Unifier #2
            M3:Money --> q q q #1:Money                 M3:Money --> q q q %1:Money
            W3:Marking --> c a #1:Money                 W3:Marking --> $ c a %1:Money
\end{verbatim}
}

\noindent
As can be seen, the use of irreducibility constraints manages to reduce the number of unifiers. By applying them to the narrowing algorithm, as canonical narrowing does, then this implies the reduction of possible steps (branches in the search tree) from the
current term, since for each one of the unifiers found between the term and the right part of a rule, we will have a new narrowing step.
\end{exa}

Additionally, to further illustrate the operation of SMT constraints, we also show an example using the built-in Maude command that invokes the Yices2 SMT solver.

\begin{exa}\label{exa:smt-maude}
We can consider two simple SMT problems to see how the check
command works. This command is used to check the satisfiability of SMT formulas in Maude, which calls an external SMT solver (Yices2) to determine it. 
It will be necessary to import the \texttt{REAL-INTEGER} module, which is a Maude module designed exclusively for handling SMT constraints, defining real and integer numbers different from those usually used in the tool. In addition, we define several variables that will be used to run the examples:

{\scriptsize
\begin{verbatim}
load smt .
fmod SMT-EXAMPLE is
    protecting REAL-INTEGER .
    vars I1 I2 I3 I4 : Integer .
    vars R1 R2 R3 R4 : Real .
endfm
\end{verbatim}
} 

\noindent
Once this is done, we load the module, and we can run our examples.
Below we define two boolean expressions, using integer and real. As we shall see, the former is satisfiable while the latter is not.

{\scriptsize
$$\begin{array}{c}
((I1\;=\;I2)\;\wedge\;(I2\;>\;I3)\;\wedge\;(I1\;\leq\;I3))\;\vee\;(I3\;\ne\;I4)
\\\\
((R1 = R2)\;\vee\;(R2 = R3))\;\wedge\;(R2\;<\;R3)\;\wedge\;(R1\;\leq R4)\;\wedge\;(R2\;>\;R4)
\end{array}$$
}

\noindent
By executing both examples with the \texttt{check} command we can see how, in fact, the results are as expected~\cite[\S 16]{maude-manual}:

{\scriptsize
\begin{verbatim}
Maude> check in SMT-EXAMPLE : 
     > (I1 === I2 and I2 > I3 and I1 <= I3) or (I3 =/== I4) .
Result from sat solver is: sat

Maude> check in SMT-EXAMPLE : 
     > (R1 === R2 or R2 === R3) and (R2 < R3) and (R1 <= R4) and (R2 > R4) .
Result from sat solver is: unsat
\end{verbatim}
} 
\end{exa}

As we will see later, each of the commands included in Maude is also found in the meta-level version, allowing the user to better control their use and the results obtained. We rely on those meta-level features to make it easier to handle both irreducibility constraints and SMT constraints. We
show in Example \ref{exa:metaCheck} below the meta-level version of the commands above.

%%%%%%%%%%%%%%%%%%%%%%%%%%%%%%%%%%%%%%%%%%%%%%%%%%%%%%%%%%%%%%%%%%%%%%%%%%%%%%%%%%%%%%%%
%%%%%%%%%%%%%%%%%%%%%%%%%%%%%%%%%%%%%%%%%%%%%%%%%%%%%%%%%%%%%%%%%%%%%%%%%%%%%%%%%%%%%%%%
\section{Implementation}
\label{sec:implementation}
The implementation of SMT canonical narrowing allows the user to choose between standard narrowing (without irreducibility constraints but with SMT constraints) or canonical narrowing
(both with irreducibility and SMT constraints).
Let us discuss the significance of the implementation. 
In Section \ref{subsec:using-meta}, 
we show how we have improved the performance of the narrowing algorithm
by using the latest meta-level versions of both equational variant unification and one narrowing step available in Maude, instead of implementing them from scratch. 
In Section \ref{subsec:data-estruc-narrow}, we describe the set of nodes used to represent the search tree, where each node includes its unique identifier as well as the identifier of its parent.
In Section \ref{subsec:search-solutions}, we show that when trying to unify the target term with every node in the search tree, the irreducibility and SMT constraints must be preserved.
In Section \ref{subsec:avoid-varclash}, we justify the use of the \texttt{\$} symbol for new variables to avoid variable clashes.
In Section \ref{subsec:algorithm-performance}, we describe the three main parts of the global narrowing algorithm and the extension of the previous data structures to gain performance.
In Section \ref{subsec:conditional-rules}, we show how we transform the SMT conditional rules into unconditional ones. 
In Section \ref{subsec:extending-algorithm}, we show how we check for satisfiability of SMT constraints in different parts of the algorithm.
In Section \ref{subsec:variable-consistency}, we describe how we keep variable consistency throughout the global narrowing algorithm. 

\begin{algorithm}[H]
\footnotesize
\caption{SMT canonical narrowing algorithm}\label{alg:smt-can-narrowing}
\begin{algorithmic}
\State \textbf{Input}: Initial Term ($\textit{ITerm}$)
\State \textbf{Input}: Target Term ($\textit{TTerm}$)
\State \textbf{Input}: Initial SMT restriction ($\textit{SMT}_0$)
\State \textbf{Input}: Initial irreducibility constraints ($\textit{Irred}_0$)
\State \textbf{Input}: Maximum depth ($\textit{MaxDepth} > 0$)
\State \textbf{Input}: Maximum solutions ($\textit{MaxSol} > 0$)
\If{($\textit{checkSAT}(\textit{SMT}_0) = \textit{true}$)}
\State $\textit{Nodes} \gets \{(\textit{ITerm},\textit{Irred}_0,\textit{SMT}_0,0)\}$
\Else
\State $\textit{Nodes} \gets \emptyset$
\EndIf
\State $\textit{Solutions} \gets \emptyset$
\While{($\textit{Nodes}\neq \emptyset\ \&\ \textit{count}(\textit{Solutions}) < \textit{MaxSol})$}
\State $\textit{Nodes} \gets \textit{Nodes} \setminus \{(\textit{Term}_i,\textit{Irred}_i,\textit{SMT}_i,\textit{Depth}_i)\}$
\If{($\textit{checkSAT}(\textit{SMT}_i) = \textit{true}$)}
    \State $\textit{Unifiers}$ $\gets$ $unify(\textit{Term}_i,\textit{TTerm},\textit{Irred}_i)$
    \State $\textit{Solutions}$ $\gets \textit{buildSolutions}(\textit{Unifiers})$
\EndIf
\If{$\textit{Depth}_i < \textit{MaxDepth}$}
\State $\textit{Children}\gets\textit{narrowingStep}(\textit{Term}_i,\textit{Irred}_i)$
    \While{$\textit{Children} \neq \emptyset$}
        \State $\textit{Children}$ $\gets$ $\textit{Children} \setminus \{(\textit{Term}',\textit{Irred}\,',\textit{SMT}\,')\}$
        \State $\textit{Nodes}$ $\gets \textit{Nodes} \cup \textit{rename}(\textit{Term}',(\textit{Irred}_i;\textit{Irred}\,'),\textit{SMT}_i\wedge\textit{SMT}\,',\textit{Depth}_i+1)$
    \EndWhile
\EndIf
\EndWhile
\State $return(Solutions)$
\end{algorithmic}
\end{algorithm}

The SMT canonical narrowing algorithm has several parts, many of them iterative,
which is outlined in
Algorithm~\ref{alg:smt-can-narrowing} above, although some details have been omitted for simplification, e.g. flags that modify the behavior of SMT constraint checking are not taken into account (see Section \ref{subsec:extending-algorithm}). 

The most important inputs for the algorithm are the initial term, the target term, the initial irreducibility constraints, the initial SMT constraint, the maximum solution limit, and the maximum depth of the tree limit. When starting the algorithm, it is checked that the initial SMT constraint is satisfied. If so, the initial node is built with the received parameters and added to the node set. Otherwise, the set of nodes will be empty, so later the algorithm will terminate without performing any further operations, giving rise to an empty solution set. After building the first node, the solution set is initialized to empty. Then the main loop of the algorithm begins. In each iteration, it is checked that the set of nodes is not empty and that the maximum number of solutions has not been reached. If this is true, one of the nodes is removed from the set, and it is verified that its SMT constraint is true. If so, an attempt is made to unify the term of the node with the objective term, and the solutions are built with the unifying results. Those solutions are added to the solution set. Afterwards, it is checked that the maximum depth has not been reached. In that case, the children of the node are generated by taking narrowing steps. After this, all the variables that contain the child nodes are renamed. Each child node will have irreducibility constraints inherited from the parent and updated, as well as the SMT constraint of the parent, to which the new constraint (if any) is added with a conjunction. After exiting the loop, the solution set is returned. Note that for each call to unification or narrowing steps, the list of irreducible terms is used.

%%%%%%%%%%%%%%%%%%%%%%%%%%%%%%%%%%%%%%%%%%%%%%%%%%%%%%%%%%%%%%%%%%%%%%%%%%%%%%%%%%%%%%%%
\subsection{Using the meta-level}
\label{subsec:using-meta}

To achieve the implementation of the new narrowing algorithm it is necessary to use some calls to the Maude meta-level available in Maude 3.3. Thanks to this, we can reuse functionalities that are integrated at the native level in C++, achieving much better performance than if we implemented them from scratch.

Each user command in Maude is represented by a corresponding command at the meta-level, allowing us greater control and management of their outputs. All meta-level commands use meta-level representations of terms. Let us briefly explain how Maude's meta-representation works with an example.

\begin{exa}\label{exa:meta-representation}
Given any term in Maude, its representation is achieved using a prefix notation in which each operator is at the head of the term, and each of the elements affected by it are enclosed in square brackets separated by commas. Each of the elements is preceded by the \texttt{'} symbol, used internally by Maude to declare identifiers. By using the operator \texttt{upTerm}, we can get the meta-representation of any term in Maude. Examples of this command are as follows

{\scriptsize
\begin{verbatim}
Maude> reduce in NARROWING-VENDING-MACHINE : upTerm(M:Marking) .
>      result Variable: 'M:Marking

Maude> reduce in NARROWING-VENDING-MACHINE : upTerm(< a c q M:Money >) .
>      result Term: '<_>['__['q.Coin,'c.Item,'a.Item,'M:Money]]
\end{verbatim}
}

\noindent Note that the meta-representation of an \texttt{M:Marking} variable is simply that same variable preceded by {'}, since the term contains no operators. However, when we introduce operators, as in the second call, the representation of the term in the meta-level is somewhat more complex, although it is achieved as explained above. The meta representation of the term \texttt{< a c q M:Money >} becomes 
\allowbreak
\texttt{`<\_>[`\_\_[`q.Coin,`c.Item,`a.Item,`M:Money]]}. The operators are simply nested and passed to the head of the term, grouping the arguments of each one of them with the square brackets.
\end{exa}

\noindent All meta-level examples below use the meta-representation of terms.

The \texttt{variant unify} command that we saw in Example \ref{exa:idem-vending-variants} corresponds to the \texttt{metaVariantUnify} command at the meta level~\cite[\S 17.6.10]{maude-manual}. It is precisely this command that we use to carry out the unification step in our implementation, since it allows us to perform equational unification modulo variant equations and axioms, as well as irreducibility constraints. The operator that defines the command is the following:

{\scriptsize
\begin{verbatim}
op metaVariantUnify :
    Module UnificationProblem TermList Qid VariantOptionSet Nat ~> UnificationPair? .
\end{verbatim}
}
The command receives six parameters and returns a structure of type \linebreak\texttt{UnificationPair?}, an error or a pair consisting of a substitution and an identifier of the family of variables used. The first command received is the module that defines the rewrite theory to work on. The second is the unification problem to which solutions are sought. The third is a list of irreducibility terms, which
are important for the canonical narrowing algorithm. 
The fourth corresponds to the identifier of the family of variables to avoid (the one used for the variables of the unification problem). The fifth is a parameter used to indicate if we want to filter the returned unifiers. Finally, a natural number parameter is received in which the unifier to be searched is indicated. We show an execution of this command in Example \ref{exa:metaVariantUnify-usage}, using in turn the module of the vending machine with idempotence
(see Example \ref{exa:idem-vending}).

\begin{exa}\label{exa:metaVariantUnify-usage}
Considering the module from Example \ref{exa:idem-vending},
we can use the \texttt{metaVariantUnify} command to find the unifiers seen in Example \ref{exa:idem-vending-variants}. We simply indicate the same equational unification problem, and by means of the last argument of the command we can select each of the unifiers to obtain. Additionally, we can use an irreducibility condition to reduce the number of unifiers just like we have seen before. For example, by using the same irreducibility condition, we can obtain one of the unifiers as follows:

{\scriptsize
\begin{verbatim}
Maude> reduce in META-LEVEL : 
>      metaVariantUnify(upModule('IDEMPOTENCE-VENDING-MACHINE, true),
>      '<_>['__['a.Item,'c.Item,'q.Coin,'M3:Money]] =? '<_>['__['$.Coin,'W3:Marking]],
>      '<_>['__['$.Coin,'M3:Money]], '@, none, 0) .
result UnificationPair: {
  'M3:Money <- '__['q.Coin,'q.Coin,'q.Coin,'#1:Money] ;
  'W3:Marking <- '__['a.Item,'c.Item,'#1:Money],'#}}
\end{verbatim}
}
\end{exa}

Another meta-level functionality that has been necessary to use is the \linebreak\texttt{metaNarrowingApply} command~\cite[\S 17.6.11]{maude-manual}. It performs a narrowing step, using the arguments shown in its definition below. Thanks to this command and the \allowbreak\texttt{metaVariantUnify} one, we can abstract from the unification processes, which are the most costly at the computational level. By invoking meta-level commands to do so, execution is done natively in C++ code, which turns out to be much faster and efficient than implementing it from scratch. The operator that defines the command is the following:

{\scriptsize
\begin{verbatim}
op metaNarrowingApply :
    Module Term TermList Qid VariantOptionSet Nat -> NarrowingApplyResult? .
\end{verbatim}
}
In this case, the command receives as the first parameter, again, the module that represents the rewrite theory to be used. The second parameter represents the term from which to perform the narrowing step. The third parameter is a list of irreducibility terms, important for canonical narrowing. The fourth parameter is the identifier of the family of variables to avoid. The fifth parameter is used to indicate if we want to filter the returned unifiers in order to get only the most general unifiers. Finally, the sixth parameter is the step that 
we want to take, that is, the ``branch" of the tree that
we want to generate from the given term. The result will be of sort \texttt{NarrowingApplyResult?}, a data structure that contains either an error, or the necessary information from the narrowing step performed.

\begin{exa}\label{exa:metaNarrowingApply-usage}
We use again the module from Example \ref{exa:idem-vending}. As an initial term we consider the 
meta-representation of (used later in the experiments) \texttt{M1:Money}. The \texttt{metaNarrowingApply} command allows us to give (among others) the first step of narrowing from that term:

{\scriptsize
\begin{verbatim}
Maude> reduce in META-LEVEL : 
>      metaNarrowingApply(upModule('IDEMPOTENCE-VENDING-MACHINE, true), 
>                   '<_>['M1:Money], empty, '@, none, 0) .
result NarrowingApplyResult: { '<_>['__['a.Item,'q.Coin,'%1:Money]],'State,
  [], 'buy-a, 'M1:Money <- '__['$.Coin,'%1:Money], 'M:Marking <- '%1:Money, '% }
\end{verbatim}
}
\noindent 
The output returned by Maude shows how the rule labeled as \texttt{buy-a} has been used to perform the narrowing step, resulting in two different assignments. On the one hand, a dollar is assigned together with a fresh variable to the variable \texttt{M1} of type \texttt{Money}. On the other hand, the same fresh variable is assigned to the variable \texttt{M2} of type \texttt{Marking} (note that in this case, this is possible only because \texttt{Money} is a subsort of \texttt{Marking}). This metalevel command also contains a parameter that, as in the Example~\ref{exa:metaVariantUnify-usage}, allows to indicate a list of irreducibility constraints. 
Consider now the situation where the second parameter is the irreducible constraint ``\texttt{M1 \$}":

{\scriptsize
\begin{verbatim}
Maude> reduce in META-LEVEL : 
>      metaNarrowingApply(upModule('IDEMPOTENCE-VENDING-MACHINE, true), 
>                   '<_>['M1:Money], '__['M1:Money,'$.Coin], '@, none, 0) .
result NarrowingApplyResult?: (failure).NarrowingApplyResult?
\end{verbatim}
}
\noindent
In this case, we see that the irreducibility constraint causes no solutions to be found. The command returns the constant \texttt{failure} to indicate this. This functionality is very important in our implementation, since it will allow us to make the calls to the narrowing steps indicating the lists of irreducibility restrictions that we accumulate at each moment.

\end{exa}

There is also a \texttt{metaNarrowingSearch} command that performs the entire narrowing algorithm instead of only one 
step. We have not used it since we need to perform intermediate operations between each narrowing 
step. Those operations handle the canonical narrowing algorithm and take SMT constraints into account.

For the SMT constraint satisfiability evaluation, there is a command in Maude that calls an external SMT solver, the \texttt{check}~\cite[\S 16.5]{maude-manual} command. Likewise, this command has its metalevel version, defined through the following operator:

{\scriptsize
\begin{verbatim}
op metaCheck : Module Term ~> Bool [special (...)] .
\end{verbatim}
}
\noindent
Both the \texttt{Integer} sort data and the \texttt{Real} sort data are supported. 
Note that 
Maude adds a new sort for Boolean SMT values, called \texttt{Boolean}, different from the 
sort for standard Boolean values, called \texttt{Bool}
(see \cite[\S 16.1]{maude-manual}).
Indeed, the response of the SMT solver will be of sort \texttt{Bool}. 

\begin{exa}
\label{exa:metaCheck}
If we consider the same module and problems used in Example \ref{exa:smt-maude}, we can try running them with the command in its meta-level version. This command will give us more information about the execution, and will return an output in a format that is much easier to handle when using it in our implementation:

{\scriptsize
\begin{verbatim}
Maude> reduce in META-CHECK : 
     > metaCheck(['SMT-EXAMPLE], 
     >      upTerm((I1 === I2 and I2 > I3 and I1 <= I3) or I3 =/== I4)) .
result Bool: (true).Bool

Maude> reduce in META-CHECK : 
     > metaCheck(['SMT-EXAMPLE], 
     >      upTerm((R1 === R2 or R2 === R3) and (R2 < R3) and (R1 <= R4) and (R2 > R4))) .
result Bool: (false).Bool
\end{verbatim}
} 

\noindent
In these outputs, a clear difference can be observed with respect to the user-level command, which returns a line in plain text as output, reporting whether the Yices2 SMT solver has determined that the problem is satisfiable or not. In this case, \texttt{Bool} values are obtained as outputs, which we can later handle 
as needed.
\end{exa}

%%%%%%%%%%%%%%%%%%%%%%%%%%%%%%%%%%%%%%%%%%%%%%%%%%%%%%%%%%%%%%%%%%%%%%%%%%%%%%%%%%%%%%%%
\subsection{Data structures and the \texttt{narrowing} command}
\label{subsec:data-estruc-narrow}

All narrowing algorithms perform one-step transitions from one
symbolic state to another ---the narrowing steps--- using the rewrite rules of the given specification. We use a tree as a data structure, in which each of these narrowing steps gives rise to a new node, with its associated term. Thus, the root node of the tree will have as its associated term the initial term (reduced to normal form) indicated by the user. At the same time, each of the nodes is itself a data structure, in which we not only find the associated term, but also some extra information that allows us to locate the node and generate new terms from it.

Our implementation is built in such a way that ten parameters are requested from the user to invoke the command, as follows:

{\scriptsize
\begin{verbatim}
narrowing(Module, Term, SearchArrow, Term, Algorithm, VariantOptionSet, TermList, Qid, 
            Bound, Bound)
\end{verbatim}
}
The first argument receives the rewrite theory to perform the unification and narrowing steps. The second and fourth arguments are used to indicate the initial term and the target term respectively. The third argument corresponds to the search arrow (\texttt{=>1}, \texttt{=>+}, \texttt{=>!}, \texttt{=>*}) that we want to use, so that solutions are included or discarded depending on the rewriting steps performed to achieve them. 
This argument may take values to indicate that only solutions that involve a single rewrite step, one or more steps, or any number of steps can be considered. The combination of the fifth and sixth parameters will indicate the type of algorithm to use. Combinations indicating the use of standard narrowing or canonical narrowing are currently accepted. The seventh argument is used to indicate a list of initial irreducibility terms to consider. 
This argument will be taken into account in all 
calls to unification procedures and in each narrowing step, allowing the value \texttt{empty} to indicate that we do not want to use irreducibility constraints in the first step. The eighth argument receives the identifier used to name the variables in the initial and target terms, to avoid later clashes. Finally, the ninth and tenth arguments are used to impose bounds on the algorithm, being able to indicate a maximum depth to expand the search tree or a maximum of solutions to search.

%%%%%%%%%%%%%%%%%%%%%%%%%%%%%%%%%%%%%%%%%%%%%%%%%%%%%%%%%%%%%%%%%%%%%%%%%%%%%%%%%%%%%%%%
\subsection{Search for Solutions}
\label{subsec:search-solutions}

When we receive the parameters from the user, the first necessary step is to verify that the value of the depth limits and solutions are admissible. If they are, the strategy to follow will be determined according to the indicated search arrow.

Once all the above is prepared, the first nodes of the search tree are generated from the root, that is, from the initial term. The tree will be generated by levels, so that children of any node belonging to the next level will not be generated until that level is completely generated. Each node contains its associated term plus some extra information. Specifically, for each node we need a unique identifier, a reference to its parent node, the branch of the tree to which it belongs, the depth to which it is located and, in the case of the use of canonical narrowing, a list of the irreducibility terms calculated so far in that branch.

Each time a new node is generated, an attempt is made to unify its associated term with the target term indicated by the user. If unifiers exist, a solution will be built for each of the unifiers found. To do this, it is necessary to go backwards through the branch to which the node belongs, combining the substitutions made to compute the accumulated substitution. If we are using canonical narrowing, when a new node is generated it will also be necessary to modify the list of irreducibility terms, adding the irreducibility term that is calculated from the irreducible left-hand side of the rule used to reach the node (see Definition \ref{def:canonical-narrowing}).

%%%%%%%%%%%%%%%%%%%%%%%%%%%%%%%%%%%%%%%%%%%%%%%%%%%%%%%%%%%%%%%%%%%%%%%%%%%%%%%%%%%%%%%%
\subsection{Avoiding variable clashes}
\label{subsec:avoid-varclash}

For the generation of new nodes, some calls are made to internal commands of the Maude meta-level. These commands only allow the indication of a variable identifier to avoid (which must be the one used previously), preventing possible variable clashes. 
Maude always chooses between three variable identifiers (\texttt{\%}, \texttt{\#} or \texttt{\@}). In each of the calls made to its internal commands, it allows us to avoid the use of one (usually the one we use in our variables, if we use one of those). But Maude can still choose between the other two. This gives rise to the possibility that variables can be repeated in different nodes, which is not an a priori problem, but it cannot be assumed when it is required to calculate the cumulative substitution of a reachability solution.

To avoid this problem, we have chosen the strategy of renaming each of the fresh variables that Maude generates on the fly, using a new identifier, the \texttt{\$} symbol. That is why in the final result returned to the user, all the fresh variables that contain the narrowing solutions will be identified with that symbol, thus ensuring that none of them clashes with the rest.

%%%%%%%%%%%%%%%%%%%%%%%%%%%%%%%%%%%%%%%%%%%%%%%%%%%%%%%%%%%%%%%%%%%%%%%%%%%%%%%%%%%%%%%%
\subsection{Algorithm performance improvement}
\label{subsec:algorithm-performance}

Due to the nature of the algorithm and the uses for which it is intended, performance of the algorithm plays a very important role. To improve this characteristic, different aspects have been taken into account regarding the sequence in which the algorithm acts and the data structures it handles.

Regarding the operators and equations in the code, they have been divided into three main parts, which correspond to the main steps of the algorithm at a theoretical level: (i) the generation of nodes (terms) in the search tree, (ii) the attempt to unify each new term with the target term, and (iii) the computation of solutions in case the unification with the target term is successful. Likewise, each of these parts is divided into subparts that facilitate not only the understanding of the code, but also a structured scheme to add new functionalities easily. Thanks to this, once we reimplemented the standard narrowing algorithm, it was relatively easy to add the new functionalities that modified it to achieve the canonical narrowing algorithm.

We can also consider the way in which the algorithm handles the data structures it works with. A priori, it could be thought that the nodes that are generated can go to a set of nodes that is subsequently processed. However, our strategy is to use an ordered list in which the nodes are processed taking into account an order similar to that of a recursion queue. In the same way, the nodes that are being processed in that list (that is, those in which the children have been generated) go to another list. This second list is used for the computation of the accumulated substitutions in the solutions. There is also another list in which the found unifiers are stored. It is also ordered to facilitate working with it recursively and calculating the solutions from the unifiers.

In addition to all this, extra parameters are dragged in the main data structure and also locally in each of the nodes. These parameters will later help to perform certain operations more quickly and efficiently. For example, each node has a reference to its parent node identifier, making it easy to go backwards on its branch if a cumulative substitution needs to be calculated.

\subsection{Using conditional rules in narrowing}
\label{subsec:conditional-rules}
To manage SMT constraints, our approach has been to use Maude's conditional rules to add them as a condition in each of the narrowing steps \footnote{\label{footnoteSMTCondition}Since Maude allows conditional rules with a condition being an SMT constraint only for rewriting, we are forced to encode our examples using a condition of the \texttt{SMTCondition = true}, where \texttt{SMTCondition} is the SMT constraint. This will  eventually be solved.}. The problem that arises is that the Maude narrowing mechanisms are not capable of processing the conditional rules. The way to fix this is to transform those conditional rules into normal rules, in which the new right-hand side of the new rules will contain both the right-hand side of the conditional rules and the SMT constraint. An operator \texttt{>>} should separate both parts, so that later the original term can be distinguished from the SMT constraint.

We have implemented a module that is responsible for carrying out the process of transformation of conditional rules. This module defines two operators: 

{\scriptsize
\begin{verbatim}
op transformMod : Module -> Module .
op transformRls : RuleSet -> RuleSet .
\end{verbatim}
}

The first receives a module, theory, module with strategy or theory with strategy. In either case, a new operator is added to the set of operators of the module or theory, which will be used to separate the terms from the SMT constraints in the transformed rules. It is also necessary to add the import of the Maude \texttt{META-TERM} module to the converted module, so that it is capable of processing the addition of this new operator. Finally, this operator calls the other defined operator, using as an argument the set of rules of the module to be transformed. For example, the equation used to transform a system module would be the following:

{\scriptsize
\begin{verbatim}
eq transformMod(mod ModId is Imports sorts Sorts . Subsorts Ops Membs Eqs Rls1 endm)
    = mod ModId is Imports (protecting 'META-TERM .)
      sorts Sorts . Subsorts 
      (Ops (op '_>>_ : 'Boolean 'State -> 'State [ctor poly (0 2)] .)) 
      Membs Eqs transformRls(Rls1) endm .
\end{verbatim}
}

The second operator, therefore, receives a set of rules, and is in charge of iterate through it looking for conditional rules. Each time a conditional rule is found, it is transformed into a new unconditional rule, in which the condition is added to the right-hand side using the \texttt{>>} operator defined above. The equations used to do this are as follows:

{\scriptsize
\begin{verbatim}
eq transformRls(Rls1 (crl Lhs => Rhs if (SMTCondition = 'true.Boolean) [Attrs].) Rls2)
    = transformRls(Rls1 Rls2) (rl Lhs => '_>>_[SMTCondition,Rhs] [Attrs narrowing] .) .
eq transformRls(Rls1) = Rls1 [owise] .
\end{verbatim}
}

If we have a conditional rule of the form \texttt{crl Lhs => Rhs if (SMTCondition = true) [Attrs]}, it will be automatically transformed into an unconditional rule of the form \texttt{rl Lhs => (SMTCondition >> Rhs) [Attrs narrowing]}, where \texttt{Lhs} and \texttt{Rhs} are variables of \texttt{Universal} type (that is, they can be instantiated as any sort), \texttt{SMTCondition} is a variable that represents the SMT constraint, and \texttt{true} is a \texttt{Boolean} value used only to be able to encode SMT constraints in the conditions of the rules (see Footnote~\ref{footnoteSMTCondition}). 
The new form of the rule after transforming it will allow us later to make the \texttt{>>} operator 
separate the term and the SMT constraint. This is explained in detail in the following section.

\subsection{Extension to handle SMT constraints}
\label{subsec:extending-algorithm}
Once we have prepared the module transformation to convert all the conditional rules into unconditional ones, we can extend the previous command so that it processes the SMT terms that will be generated with the new rules. This extension has been done without making changes at the user level, except for the addition of possible values to one of the existing arguments, as well as a new argument that allows for indicate an initial SMT constraint:

{\scriptsize
\begin{verbatim}
narrowing(Module, Term, SearchArrow, Term, AlgorithmOptionSet, VariantOptionSet, TermList, 
            Term, Qid, Bound, Bound)
\end{verbatim}
}

Until now, the fifth argument, of type \texttt{AlgorithmOptionSet}, only accepted the \texttt{standard} and \texttt{canonical} values, used to indicate the type of narrowing algorithm to use. Now, it also accepts combinations of those two values with the \texttt{smt}, \texttt{noCheck} and \texttt{finalCheck} values, although the second and third are limitations of the first, so they cannot appear without it. Therefore, there are now various combinations of accepted values for the parameter: (i) \texttt{standard}, (ii) \texttt{canonical}, (iii) \texttt{smt standard}, (iv) \texttt{smt canonical}, (v) \texttt{smt noCheck standard}, (vi) \texttt{smt noCheck canonical}, (vii) \texttt{smt finalCheck standard} and (viii) \texttt{smt finalCheck canonical}.

By using the \texttt{smt} value, the transformation of the conditional rules will be performed in the module used as rewrite theory for each conditional rule with an SMT constraint. The common approach is that the accumulated SMT constraint will be processed during the execution of the algorithm to check if it is satisfiable each time a new SMT constraint is added to the accumulated SMT constraint. By using \texttt{smt noCheck}, only the transformation of the rules will be carried out, ignoring the satisfiability of the accumulated SMT constraint. On the other hand, by using \texttt{smt finalCheck}, the accumulated SMT constraint will be carried along each narrowing sequence, but it will be evaluated for satisfiability only at the end, when the solutions with the target term are already computed, discarding those in which the SMT constraint is not satisfied.

The most relevant changes to the algorithm occur before trying to unify the term of a new generated node with the target term, since the satisfiability of the SMT constraint for that node will have to be checked first if a conditional rule was applied. The SMT constraint has no impact before that, so the narrowing steps can be done just as they were before. The difference is that the term that is generated in each step may contain an SMT restriction, which is exactly what is checked later. 

Furthermore, we need to modify the previously used data structures. Now the main structure must save the initial SMT constraint indicated by the user. It will also be necessary for each of the nodes to contain 
an SMT constraint (which can in turn be a conjunction of SMT constraints). We have stored that constraint at each node in a \texttt{\{Term, Boolean\}} pair, where the second value of the pair indicates the satisfiability of the constraint found in the first value. Two new operators are introduced in the algorithm that run after the generation of a new node and renaming of its variables, although they will only be used if the user indicates that SMT constraints must be processed:

{\scriptsize
\begin{verbatim}
op evaluateSMT : UserArguments TreeInfo SolutionList -> NarrowingInfo .
op checkSat : UserArguments TreeInfo SolutionList -> NarrowingInfo .
\end{verbatim}
}

The \texttt{evaluateSMT} operator performs the separation of the SMT constraints from the new term generated with one of the transformed rules. In turn, it joins these restrictions with the SMT conjunction of restrictions carried so far, which will come from the predecessor nodes to the current one and from the initial restrictions indicated by the user. Additionally, it evaluates all those restrictions, to know if they are satisfiable or not. To do this, we rely on Maude's SMT interface, which is available in the meta-level. Specifically, we use the \texttt{metaCheck} command (detailed above in Example \ref{exa:metaCheck}), which receives the module to use and the term to evaluate, returning a value of type \texttt{Bool}. If the result is \texttt{true}, the constraints are satisfiable. Otherwise, \texttt{false} is returned. Note that in case the user has indicated, in addition to the \texttt{smt} value as an argument, the \texttt{noCheck} or \texttt{finalCheck} value, the \texttt{evaluateSMT} operator will only separate the SMT constraint from the term, ignoring the rest of the process, since we are not interested in checking the satisfiability, although in the second case it will be necessary to evaluate the constraints later.

The \texttt{checkSat} operator is responsible for processing the result obtained when executing the \texttt{metaCheck} function. If the restrictions are satisfiable, the next execution step should be the attempt to unify the term of the node with the objective term, to check if it corresponds to one or more solutions of the reachability problem. If the constraints are not satisfiable, then it will not make sense to perform the unification step, since we will not consider the term of the node as valid. We therefore return to the step of generating new nodes, marking the current node as invalid, so that it is not taken into account later, since we do not want to generate the possible child nodes of this node either.

\subsection{Variable consistency}
\label{subsec:variable-consistency}
As we explained in Section \ref{subsec:avoid-varclash}, the way Maude generates the fresh variables may lead to clashes. For this reason, the fresh variables that are generated in each narrowing step must be renamed. Since the variables in an SMT constraint are related to the term of tree node, as well as to the variables in the previously accumulated SMT constraints, there is a consistency problem with this renaming. That is why in each narrowing step, we now have to apply variable substitutions to the accumulated SMT constraint so that there is no such loss of consistency. Specifically, at each narrowing step, the computed substitution that must be applied to the term of the previous node to take that step must be applied to the new node's SMT constraint. The substitution must also be applied to the different SMT constraints carried along the node branch. In turn, this sequence of SMT constrains will already come with the variables renamed in the previous steps, so 
variable coherence is met. Note that the initial SMT constraint indicated by the user will also have to be renamed. This is not a problem, since that constraint is also automatically added to the accumulated SMT constraint of the initial node, so it can be renamed at the same time as the rest.

%%%%%%%%%%%%%%%%%%%%%%%%%%%%%%%%%%%%%%%%%%%%%%%%%%%%%%%%%%%%%%%%%%%%%%%%%%%%%%%%%%%%%%%%
%%%%%%%%%%%%%%%%%%%%%%%%%%%%%%%%%%%%%%%%%%%%%%%%%%%%%%%%%%%%%%%%%%%%%%%%%%%%%%%%%%%%%%%%
\section{Experiments}
\label{sec:experiments}
We have performed several experiments using different rewrite theories. 
The reimplementation of both the standard narrowing and canonical narrowing in the same command presented in Section~\ref{sec:implementation} allows us to perform more faithful comparisons between the algorithms, independently of the standard built-in narrowing algorithm provided by Maude at the C++ level. However, since the built-in narrowing returns only one solution when executed via its meta-level function, we have also built a command that iteratively obtains all solutions. In summary, for the examples in this section, we actually include a comparison between (i) the standard built-in narrowing, (ii) our implementation of standard narrowing, and (iii) our implementation of canonical narrowing. 

Note that as we progress through this section, we present experiments of increasing complexity. 
Sections \ref{subsec:vending-machine}, \ref{subsec:xor-protocol}, \ref{subsec:proc-counter} and \ref{subsec:idempotence-vending-machine}
have no conditional rules but each one has an equational theory more complex than the previous one. 
Sections \ref{subsec:bank-account}, \ref{subsec:brands-chaum-time} and \ref{subsec:brands-chaum-time-space}
have conditional rules with SMT constraints on real arithmetic 
but each one has an equational theory more complex than the previous one. 

Let us emphasize that
the experiments of Section~\ref{subsec:brands-chaum-time-space} 
compare 
standard narrowing (without irreducibility constraints but with SMT constraints) and canonical narrowing
(both with irreducibility and SMT constraints) 
combining: 
(i) SMT constraints on non-linear real arithmetic, 
(ii) the exclusive-or theory,
and
(iii) the commitment theory between the participants.

All of the Maude modules and experiments are available at \url{https://github.com/ralorueda/wrla22-jlamp}.

%%%%%%%%%%%%%%%%%%%%%%%%%%%%%%%%%%%%%%%%%%%%%%%%%%%%%%%%%%%%%%%%%%%%%%%%%%%%%%%%%%%%%%%%
\subsection{Vending machine}
\label{subsec:vending-machine}
In the first set of experiments we  use a module that is a classic in the Maude community. It is the coffee and apple vending machine, in which dollars and quarters are inserted to buy combinations of those products. To do this, we specify that each coffee costs one dollar and each apple three-quarters of a dollar. Two rules handle state transitions for those specifications. Furthermore, an equation is used to specify the change of four-quarters of a dollar to one dollar. Note the addition of a variable \texttt{M} of type \texttt{Marking} to make the rules and equations ACU-coherent \cite{DBLP:journals/tcs/Meseguer17}.

{\scriptsize
\begin{verbatim}
mod NARROWING-VENDING-MACHINE is
   sorts Coin Item Marking Money State .
   subsort Coin < Money .
   op empty : -> Money .
   op __ : Money Money -> Money [assoc comm id: empty] .
   subsort Money Item < Marking .
   op __ : Marking Marking -> Marking [assoc comm id: empty] .
   op <_> : Marking -> State .
   ops $ q : -> Coin .
   ops c a : -> Item .
   var M : Marking .
   rl [buy-c] : < M $ > => < M c > [narrowing] .
   rl [buy-a] : < M $ > => < M a q > [narrowing] .
   eq [change] : q q q q M = $ M [variant] .
endm
\end{verbatim}
}

We use the reachability problem $< M1 >\;\leadsto^{*}_{\alpha,R/E,B}\; St $ where \texttt{M1} is a variable of type \texttt{Money} and \texttt{St} is a variable of type \texttt{State}. That is, we are asking for all the states that can be reached from an initial state containing only \texttt{quarters} and \texttt{dollars}. It is a fairly generic problem that allows us to see the number of nodes that are being generated in the search tree. Table \ref{tab:experiments-vending} shows the results of running the command with this reachability problem.

\begin{table}[ht]
    \centering
    \setlength\doublerulesep{0.5pt}
    \caption{Experiments using the vending machine module.}
    \label{tab:experiments-vending}
    \begin{tabular}{ |c|c|c|c| } 
        \hline
        Algorithm & Depth limit & Execution time & Solutions found    \\
        \hhline{====}
        Native      & 4           & 24 ms             & 163           \\
        \hline
        Standard    & 4           & 56 ms             & 163           \\
        \hline
        Canonical   & 4           & 44 ms             & 137           \\
        \hhline{====}
        Native      & 5           & 88 ms             & 550           \\
        \hline
        Standard    & 5           & 436 ms            & 550           \\
        \hline
        Canonical   & 5           & 260 ms            & 119           \\
        \hhline{====}
        Native      & 6           & 380 ms            & 1850          \\
        \hline
        Standard    & 6           & 6056 ms           & 1850          \\
        \hline
        Canonical   & 6           & 2888 ms           & 1213          \\
        \hhline{====}
        Native      & 7           & 2364 ms           & 6216          \\
        \hline
        Standard    & 7           & 184536 ms         & 6216          \\
        \hline
        Canonical   & 7           & 50420 ms          & 3559          \\  
        \hline
    \end{tabular}
\end{table}

These initial experiments use a simple rewrite theory and a simple reachability problem. As a consequence, the narrowing included natively in Maude turns out to be faster than either of our two algorithms, thanks to its coding in C++. However, we can see that even in these cases, if we compare our standard narrowing implementation with our canonical narrowing implementation, the latter has always a better performance. This leads us to think that a natively programmed canonical narrowing would be able to outperform Maude's standard narrowing even using these simple parameters. To strengthen this idea, we can look at the number of solutions (which in this case represent the number of states in the tree) found. For example, for depth level 7, canonical narrowing is capable of reducing the number of states generated by almost half regarding standard narrowing. If it was implemented natively in Maude, its execution time would obviously be much less, since it has to go through far fewer rewriting steps. 

In addition, the reduction of solutions represents in itself a relevant 
improvement. This reduction of solutions is given by the elimination of redundant narrowing steps, that is narrowing steps whose computed substitution is not normalized and, by confluence and coherence between rules and equations, the normalized version of the computed substitution will correspond to a different narrowing step.

%%%%%%%%%%%%%%%%%%%%%%%%%%%%%%%%%%%%%%%%%%%%%%%%%%%%%%%%%%%%%%%%%%%%%%%%%%%%%%%%%%%%%%%%
\subsection{A protocol using the exclusive-or property}
\label{subsec:xor-protocol}
For the second set of experiments we use a protocol using some strand space notation (see \cite{ESCOBAR2006162}), that includes the exclusive-or property as equational theory. This theory is shown below. Note the addition of the second equation for AC-coherence.

{\scriptsize
\begin{verbatim}
fmod EXCLUSIVE-OR is 
  sort XOR .  
  op mt : -> XOR .
  op _*_ : XOR XOR -> XOR [assoc comm] .
  vars X Y Z U V : [XOR] .
  eq [idem] :     X * X = mt    [variant] .
  eq [idem-Coh] : X * X * Z = Z [variant] .
  eq [id] :       X * mt = X    [variant] .
endfm
\end{verbatim}
}

\noindent
In the following XOR-PROTOCOL module, 
the equational theory above is imported and the rest of the protocol is implemented. The main structure is a state that stores the set of messages that have been sent and the new messages to be sent. The exchange of messages is done between two users for the protocol to take place. The \texttt{-} and \texttt{+} symbols are used as operators to distinguish between the messages to be received or sent respectively. The nonce generation is included in the protocol, as well as data structures that specify the knowledge that an intruder might have.

{\scriptsize
\begin{verbatim}
mod XOR-PROTOCOL is protecting EXCLUSIVE-OR .
  sorts Name Nonce Fresh Msg . subsort Name Nonce XOR < Msg . subsort Nonce < XOR .
  ops a b c : -> Name . op n : Name Fresh -> Nonce .
  op pk : Name Msg -> Msg . ops r1 r2 r3 : -> Fresh .
  sort SMsg . sort SMsgList . subsort SMsg < SMsgList .
  ops + - : Msg -> SMsg .
  op nil : -> SMsgList .
  op _`,_ : SMsgList SMsgList -> SMsgList [assoc] .

  sort Strand . sort StrandSet . subsort Strand < StrandSet .
  op `[_|_`] : SMsgList SMsgList -> Strand .
  op mt : -> StrandSet .
  op _&_ : StrandSet StrandSet -> StrandSet [assoc comm id: mt] .

  sort IntruderKnowledge .
  op mt : -> IntruderKnowledge .
  op inI : Msg -> IntruderKnowledge .
  op nI : Msg -> IntruderKnowledge .
  op _`,_ : IntruderKnowledge IntruderKnowledge -> IntruderKnowledge [assoc comm id: mt] .

  sort State .
  op Sta : -> State .
  op `{_`{_`}`} : StrandSet IntruderKnowledge -> State .

  vars IK IK1 IK2 : IntruderKnowledge .  vars A B : Name .  
  vars X Y Z U V : [XOR] .  vars SS SS1 SS2 : StrandSet .  
  var M : Msg .  vars L1 L2 : SMsgList .  vars NA NB : Nonce .

  rl [r1] : { (SS & [ ( L1 , -(M)) | L2 ])  { (inI(M) , IK) } } =>
            { (SS & [ L1 | (-(M) , L2) ])  { (inI(M) , IK) } } 
            [narrowing] .
  rl [r2] : { (SS & [ (L1 , +(M)) | L2 ])  { (inI(M) , IK) } } =>
            { (SS & [ L1 | (+(M) , L2) ])  { (nI(M) , IK) } } 
            [narrowing] .
endm
\end{verbatim}
}

We can define a reachability problem by using a basic message exchange between users. To do this, we consider a backwards execution, so that the target term will be the initial state of the message stack, while the initial term will be the final state. If a solution is found, it means that that execution trace exists, so it could occur in the protocol. The Maude-NPA \cite{DBLP:conf/fosad/EscobarMM07} tool works in a similar way to this. 

Considering \texttt{X} and \texttt{Y} as variables of type \texttt{Msg}, the reachability problem that we have defined for the experiments is the following:

{\scriptsize
$$\begin{array}{c}
\{[nil,\;+(pk(a, n(b, r1))),\;-(pk(b, Y)),\;+(Y * n(b, r1))\;|\;nil] \\\&\;[nil,\;-(pk(a, X)),\;+(pk(b, n(a,r2))),\;-(X * n(a, r2))\;|\;nil]\\\{inI(X * n(a, r2)),\;inI(pk(a, X)),\;inI(pk(b, Y))\}\} 
\\
\stackrel{?}{\leadsto}{\!}^{*}\\
\{[nil\;|\;+(pk(a, n(b, r1))),\;-(pk(b, Y)),\;+(Y * n(b, r1)),\;nil] \\\&\; [nil\;|\;-(pk(a, X)),\;+(pk(b, n(a, r2))),\;-(X * n(a, r2)),\;nil]\\\{nI(X * n(a, r2)),\;nI(pk(a,X)),\;nI(pk(b, Y))\}\}
\end{array}$$
}

\noindent
In this context we are working with a finite search space, so we can ignore the limit of solutions and the depth limit (although all the solutions are found in depth 6, so we could also use that depth limit). Results are shown in Table \ref{tab:experiments-xor}.

\begin{table}[ht]
    \centering
    \setlength\doublerulesep{0.5pt}
    \caption{Experiments using the XOR-protocol module.}
    \label{tab:experiments-xor}
    \begin{tabular}{ |c|c|c|c| } 
        \hline
        Algorithm  & Execution time    & Solutions found \\
        \hhline{====}
        Native     & 1492 ms           & 84              \\
        \hline
        Standard   & 14168 ms          & 84              \\
        \hline
        Canonical  & 2168 ms           & 1               \\
        \hline
    \end{tabular}
\end{table}

In this case, the native standard narrowing in Maude again 
is faster than our two algorithms, although the difference is less than before. The usual impact that canonical narrowing has on the number of returned solutions is further noticeable. As we mentioned before, thanks to carrying out a ``pruning" of the tree by discarding those reducible branches, canonical narrowing is able to reduce the 84 initial solutions to only 1. 

It is important to note here the usefulness of canonical narrowing in the field of security protocols, and specifically for
tools relying on unification and/or narrowing, such as the Maude-NPA tool \cite{DBLP:conf/fosad/EscobarMM07},
Tamarin \cite{DBLP:conf/cav/MeierSCB13} and AKISS \cite{DBLP:journals/tocl/ChadhaCCK16}.
By 
eliminating redundant narrowing steps, as seen in Example \ref{exa:idem-vending-variants}, when we analyze a protocol with canonical narrowing we achieve higher performance with less numerous reachable states but still complete results. 

We have already seen that canonical narrowing is useful even when ---due to the prototype nature of its present implementation--- it cannot be faster than the C++ based native standard narrowing in Maude. We have also concluded that if it were also integrated natively, it could be substantially faster and generate fewer states than standard narrowing in many cases. But if we also carry out experiments with systems in which there are many variants, these claims will be further reinforced, as shown in the next section.

%%%%%%%%%%%%%%%%%%%%%%%%%%%%%%%%%%%%%%%%%%%%%%%%%%%%%%%%%%%%%%%%%%%%%%%%%%%%%%%%%%%%%%%%
\subsection{A simplified process counter}
\label{subsec:proc-counter}
For these experiments, we first implement a simple module defining the properties of an abelian group. That will be the equational theory used. This equational theory is complex for narrowing, due to the number of variants \cite{DBLP:journals/corr/abs-1909-08241, DBLP:journals/corr/abs-2009-11070}. That is why with these experiments we can see the real potential of canonical narrowing compared to standard narrowing.

{\scriptsize
\begin{verbatim}
fmod ABELIAN-GROUP is
   sort Int .
   ops 0 1 : -> Int [ctor] .
   op _+_ : Int Int -> Int [assoc comm prec 30] .
   op -_ : Int -> Int .
   vars X Y Z : Int .

   eq X + 0 = X [variant] .
   eq X + (- X) = 0 [variant] .
   eq X + (- X) + Y = Y [variant] .
   eq - (- X) = X [variant] .
   eq - 0 = 0 [variant] .
   eq (- X) + (- Y) = -(X + Y) [variant] .
   eq -(X + Y) + Y = - X [variant] .
   eq -(- X + Y) = X + (- Y) [variant] .
   eq (- X) + (- Y) + Z = -(X + Y) + Z [variant] .
   eq - (X + Y) + Y + Z = (- X) + Z [variant] .
endfm
\end{verbatim}
}

\noindent
A rewrite theory which consists of a pair of integers that function as process counter is defined (borrowed from \cite{DBLP:conf/birthday/AlpuenteBFS14}). The first integer represents the processes that are running, while the second represents those that have finished their execution. The transition rule represents the termination of a process that was running, so that the value of the first integer of the pair is decreased by one, and at the same time the value of the second integer is increased by one. The transition rule allows for narrowing, while the abelian group equations allow for the generation of variants. 
By combining the two theories, we get a system of transitions that, as mentioned above, turns out to be quite complex, due to the large number of variants that any term will normally have (see \cite{DBLP:journals/corr/abs-1909-08241,DBLP:journals/corr/abs-2009-11070}).

{\scriptsize
\begin{verbatim}
mod PROC-COUNTER is protecting ABELIAN-GROUP .
   sort State .
   op <_,_> : Int Int -> State [ctor] .
   vars X Y Z : Int .

   rl [finish-proc] : < (X + 1),Y > => < ((X + 1) + (- 1)),(Y + 1) > [narrowing] .
endm
\end{verbatim}
}

Considering \texttt{X} and \texttt{Y} 
as variables of type \texttt{Int}, we use a common initial term: $< 0,\,1+X >$. The target term will vary slightly allowing us to fix the depth at which we want to find the solution. For depth one, it will be $<-1,\,Y>$. For depth two, it will be $< -(1+1),\,Y >$. For depth three, it will be $< -(1+1+1),\,Y >$, and for depth four, it will be $< -(1+1+1+1),\,Y >$.

Apparently, the initial term and the target terms are very simple in this example, but the computation becomes very complex, since the search tree will grow very quickly in width. Table \ref{tab:abelian-group} shows the results of executing the above problems using different depth limits. 

\begin{table}[ht]
    \centering
    \setlength\doublerulesep{0.5pt}
    \caption{Experiments using the process counter module.}
    \label{tab:abelian-group}
    \begin{tabular}{ |c|c|c|c| } 
        \hline
        Algorithm   & Depth limit & Execution time    & Solutions found \\
        \hhline{====}
        Native      & 1           & 300 ms            & 184             \\
        \hline
        Standard    & 1           & 316 ms            & 184             \\
        \hline
        Canonical   & 1           & 312 ms            & 184             \\
        \hhline{====}
        Native      & 2           & $>$ 10 h           & $-$             \\
        \hline
        Standard    & 2           & $>$ 10 h           & $-$             \\
        \hline
        Canonical   & 2           & 1948 ms           & 719             \\
        \hhline{====}
        Native      & 3           & $>$ 10 h           & $-$             \\
        \hline
        Standard    & 3           & $>$ 10 h           & $-$             \\
        \hline
        Canonical   & 3           & 9192 ms          & 2033            \\
        \hhline{====}
        Native      & 4           & $>$ 10 h           & $-$             \\
        \hline
        Standard    & 4           & $>$ 10 h           & $-$             \\
        \hline
        Canonical   & 4           & 36760 ms          & 4969            \\
        \hline
    \end{tabular}
\end{table}

In this case, we can see how the execution time for the first level (i.e., first reachability problem) is very similar in any of the three algorithms. Furthermore, it is striking that the solutions returned are the same. This is normal, since canonical narrowing does not have any kind of impact on the first level, because it has not yet calculated (see Definition \ref{def:canonical-narrowing}) irreducibility constraints (unless we specify them as part of the initial call). However, we can see how from depth 2, our standard narrowing algorithm does not even manage to finish in a reasonable execution time. The built-in narrowing 
does not do it either. In contrast, the canonical narrowing algorithm does terminate, returning a large number of solutions in a relatively short time. The deeper we go into the tree, the more solutions are found. At the same time, the execution time grows, but it is still acceptable.

This is a clear example of the enormous improvement that canonical narrowing can bring over standard narrowing in many cases, even using the one found natively in Maude. Obviously, if we put canonical narrowing at the same level, that is, included in Maude natively, the performance difference would be extremely large in favor of canonical narrowing, especially in this type of cases.

%%%%%%%%%%%%%%%%%%%%%%%%%%%%%%%%%%%%%%%%%%%%%%%%%%%%%%%%%%%%%%%%%%%%%%%%%%%%%%%%%%%%%%%%
\subsection{Idempotence vending machine}
\label{subsec:idempotence-vending-machine}
As we mentioned earlier in the introduction, idempotence is a very important property in computing, since, for example,
set data types enjoy it. Canonical narrowing seems to behave very well managing this property when compared to standard narrowing (even better than the experiments with an abelian group). We have done some experiments in which this property is used to corroborate this. We use the vending machine module with idempotence on items and dollars (see Example \ref{exa:idem-vending}). The reachability problem defined in this case is $< M1 >\;\leadsto^{*}_{\alpha,R/E,B}\; < \$\;a\;c\;M2 > $, where \texttt{M1} is a variable of type \texttt{Money} and \texttt{M2} is a variable of type \texttt{Marking}.

\begin{table}[ht]
    \centering
    \setlength\doublerulesep{0.5pt}
    \caption{Experiments using the vending machine with idempotence of Example~\ref{exa:idem-vending}}
    \label{tab:idempotence-vending}
    \begin{tabular}{ |c|c|c|c| } 
        \hline
        Algorithm   & Depth limit & Execution time    & Solutions found \\
        \hhline{====}
        Native      & 4           & 1908 ms           & 3804             \\
        \hline
        Standard    & 4           & 5792 ms            & 3804             \\
        \hline
        Canonical   & 4           & 608 ms            & 856              \\
        \hhline{====}
        Native      & 6           & 116144 ms          & 40284            \\
        \hline
        Standard    & 6           & 1185464 ms         & 40284            \\
        \hline
        Canonical   & 6           & 8888 ms          & 4284             \\
        \hhline{====}
        Native      & 8           & $>$ 10 h          & $-$              \\
        \hline
        Standard    & 8           & $>$ 10 h          & $-$              \\
        \hline
        Canonical   & 8           & 39772 ms         & 18963            \\
        \hline
    \end{tabular}
\end{table}

Table \ref{tab:idempotence-vending} shows the results obtained when using the reachability problem. We must bear in mind that in this case, once again, the growth of the tree in width is very large, due to the large number of variants of the system.

In this case we can see that the executions with a lower depth limit are extremely fast. However, even in those cases the difference is obvious, with the canonical narrowing being faster and returning fewer solutions. As we increase the depth limit, the difference in performance becomes more and more noticeable.

Just by looking at the experiments with depth limit 4, we can see that canonical narrowing achieves a performance at the computational level about 14 times better than Maude's native standard narrowing. And if we instead look at our Maude reimplementation of standard narrowing, for a fair comparison, the difference is huge. More than 15 minutes of execution is reduced to just 7 seconds. The number of solutions returned is also reduced to one tenth, showing that the large percentage of those calculated by standard narrowing were unnecessary. For depth limit 5, the execution times of the standard narrowing are no longer reasonable, while the canonical narrowing manages to finish in less than 4 minutes.

%%%%%%%%%%%%%%%%%%%%%%%%%%%%%%%%%%%%%%%%%%%%%%%%%%%%%%%%%%%%%%%%%%%%%%%%%%%%%%%%%%%%%%%%
\subsection{Bank account with SMT constraints}
\label{subsec:bank-account}
As mentioned above, in addition to the canonical narrowing implementation, we have extended the algorithm to be able to handle conditional rewrite theories with SMT constraints. Many of them are theories that 
could not run with narrowing before. For example, \cite{DBLP:journals/jlap/Meseguer20} presents a module that models the basic behavior of a bank account. To do this, various operations with natural numbers are defined. The original module cannot be run with narrowing on Maude, as it contains conditional rules. In the article, a manual transformation is proposed to eliminate variants and conditional rules, so that the module can be used, among other things, to run with narrowing. But if we specify the conditional rule conditions as SMT constraints, using reals instead of naturals (which is even closer to the real bank account approach), our implementation is able to process the original specification. The transformed module is as follows:

{\scriptsize
\begin{verbatim}
mod BANK-ACCOUNT is 
    protecting REAL-INTEGER .
    protecting TRUTH-VALUE .

    sorts Account Msg MsgConf State StatePair .
    subsort  Msg < MsgConf .

    op < bal:_pend:_overdraft:_> : Real Real Bool -> Account [ctor] . 
    op mt : -> MsgConf [ctor] .
    op withdraw : Real -> Msg [ctor] .
    op _,_ : MsgConf MsgConf -> MsgConf [ctor assoc comm id: mt] . 
    op _#_ : Account MsgConf -> State [ctor] .  ***  state ctor

    vars n n' m x y x' y' : Real .    vars b b' : Bool .
    vars s s' : State . var msgs : MsgConf .

    rl [w-req] : < bal: n + m + x pend: x overdraft: false > # msgs => 
                    < bal: n + m + x pend: x + m overdraft: false > # withdraw(m),msgs 
                            [narrowing] .

    rl [dep] :  < bal: n pend: x overdraft: false > # msgs =>
                    < bal: n + m pend: x overdraft: false > # msgs [narrowing nonexec] .

    crl [w1] :  < bal: n pend: x overdraft: false > # withdraw(m),msgs => 
                    < bal: n pend: x overdraft: true > # msgs  
                        if (m > n) = true .

    crl [w2] :  < bal: n pend: x overdraft: false > # withdraw(m),msgs => 
                    < bal: (n - m) pend: (x - m) overdraft: false > # msgs  
                        if (m <= n) = true .
endm
\end{verbatim}
}

Now we can test its performance using different reachability problems, using our narrowing algorithm. For example, we can use a reachability problem in which all elements are variables. That is, the most generic possible reachability problem:

{\scriptsize
$$\begin{array}{c}
<\;bal{:}\;X{:}Real\;pend{:}\;Y{:}Real\;overdraft{:}\;B{:}Bool\;>\;\#\;M{:}MsgConf
\\
\stackrel{?}{\leadsto}{\!}^{*}\\
<\;bal{:}\;X'{:}Real\;pend{:}\;Y'{:}Real\;overdraft{:}\;B'{:}Bool\;>\;\#\;M'{:}MsgConf
\end{array}$$
}

By varying the maximum number of levels to display in the search tree, we obtain the results shown in Table 3.

\begin{table}[ht]
    \centering
    \setlength\doublerulesep{0.5pt}
    \caption{Experiments using the bank account module and SMT narrowing.}
    \label{tab:bank-account}
    \begin{tabular}{ |c|c|c|c| } 
        \hline
        Depth limit & Execution time    & Solutions found \\
        \hhline{====}
        3           & 20 ms             & 49             \\
        \hline
        4           & 88 ms             & 134             \\
        \hline
        5           & 460 ms            & 360             \\
        \hline
        6           & 4304 ms           & 976             \\
        \hline
        7           & 66724 ms          & 2664            \\
        \hline
    \end{tabular}
\end{table}

These results show the performance of narrowing for conditional modules with SMT constraints and no variant equations. It can be seen that for various depths of the search tree, the computation time is reasonable but still exponential (inevitable), even for the very generic reachability problem we use (for more instantiated problems, the performance should be even better). As expected, the computation time 
increases a lot along with the solutions going further down the tree. However, we will see later that for modules with variants, canonical narrowing will help alleviate this by reducing state explosion.

%%%%%%%%%%%%%%%%%%%%%%%%%%%%%%%%%%%%%%%%%%%%%%%%%%%%%%%%%%%%%%%%%%%%%%%%%%%%%%%%%%%%%%%%
\subsection{Brands and Chaum with time}
\label{subsec:brands-chaum-time}
So far, on the one hand, we have shown examples where we compare the performance of standard narrowing and canonical narrowing. On the other hand, we have tested the performance of standard narrowing with conditional rules that include SMT constraints. But our implemented command goes further, allowing conditional rewrite theories with SMT constraints to be processed using both standard and canonical narrowing. Next, we show an example in which, using a conditional module with SMT constraints, we compare both narrowing algorithms.

We rely on the generic rewrite theory for protocol specification, inspired on the strand spaces~\cite{strands} used by Maude-NPA~\cite{DBLP:conf/fosad/EscobarMM07}, but with some modifications that adapt it to include SMT constraints on the real numbers, inspired on the constraints used in \cite{DBLP:conf/indocrypt/Aparicio-Sanchez20,DBLP:conf/birthday/Aparicio-Sanchez21}. It is a module that allows us to specify a state, made up of sets of strands and the intruder knowledge, which represents the communication channel. With it we can represent the protocols in a generic way, adding the corresponding equational theories for each of them. Later, when coding the narrowing calls, we will specify the exact strands of each protocol. 

In the original module (see Section~\ref{subsec:xor-protocol}), we had two transition rules. One of them (labeled as \texttt{send-msg}) processes the sent messages, and the other (labeled as \texttt{receive-msg}) the received messages:

{\scriptsize
\begin{verbatim}
var IK : IntruderKnowledge .   var SS : StrandSet .   var M : Msg .   vars L1 L2 : SMsgList .

rl [receive-msg] : { (SS & [ ( L1 , -(M)) | L2 ])  { (inI(M) , IK) } } =>
                    { (SS & [ L1 | (-(M) , L2) ])  { (inI(M) , IK) } } [narrowing] .
                    
rl [send-msg] : { (SS & [ (L1 , +(M)) | L2 ])  { (inI(M) , IK) } } =>
                    { (SS & [ L1 | (+(M) , L2) ])  { (nI(M) , IK) } } [narrowing] .
\end{verbatim}
}

\noindent
It can be seen in each of them how, for each set of strands, represented in square brackets, there is a list to the left of the operator \texttt{|} and one to the right. The first contains the messages to be processed, while the second contains the processed messages. At each transition, a message (sent or received) is taken from the end of the list of messages to be processed and moved to the list of processed messages. In the event that it is a sent message, the correspondence of that message will also be modified in the communication channel or intruder knowledge.

To adapt the module to protocols using non-linear arithmetic constraints on the real numbers via satisfiability, we add a conditional rule that is responsible for processing a new type of data that can appear in the strands sets: constraints. Specifically in our case, SMT constraints (type \texttt{Boolean}), which will be represented in the channel between the messages with the operator \texttt{\{\_\}}. We will therefore now have three rules. One of them is responsible for processing the messages sent, another the messages received, and another the restrictions that occur at any given time:

{\scriptsize
\begin{verbatim}
var IK : IntruderKnowledge .   var SS : StrandSet .   var SSR : StrandSetR .
var SSN : StrandSetN .   var M : Msg .   vars LeE2 : SMsgList-eE .   
var LREe1 : SMsgListR-Ee .

crl [check-contraint] : { (SSR & [ LREe1 , {B:Boolean} | LeE2 ])  { IK } } =>
                    { (SSR & [ LREe1 | {B:Boolean} , LeE2 ])  { IK } } 
                        if B:Boolean = true [nonexec] .
rl [receive-msg] : { (SSN & [ LREe1 , -(M) | LeE2 ])  { (inI(M) , IK) } } =>
                    { (SSN & [ LREe1 | -(M) , LeE2 ])  { (inI(M) , IK) } } [narrowing] .
rl [send-msg] : { (SS & [ LREe1 , +(M) | LeE2 ])  { (inI(M) , IK) } } =>
                    { (SS & [ LREe1 | +(M) , LeE2 ])  { (nI(M) , IK) } } [narrowing] .
\end{verbatim}
}

\noindent
Note that in this case we use variables from different sorts, \texttt{SMsgListR-Ee} and \texttt{SMsgList-eE}, rather than the ones we used in \cite{DBLP:conf/wrla/Lopez-Rueda22}. 
We were interested in replicating 
the partial order reduction technique of Maude-NPA that establishes priorities between actions \cite{DBLP:journals/iandc/EscobarMMS14}. 
First, the conditional rule 
\texttt{check-constraint} is attempted. If not applicable, the unconditional rule
\texttt{receive-msg} is attempted. If not applicable, the unconditional rule
\texttt{send-msg} is attempted. The result is the elimination of the indeterminism, obtaining a much faster narrowing execution thanks to the improvement in the protocol specification. 

The previous module allows us, in a generic way, to specify protocols that contain SMT constraints. 
Now we just need to add the specific equational theories of each protocol we want to model. In our case, the first protocol used is Brands and Cham with time \cite{DBLP:conf/indocrypt/Aparicio-Sanchez20}, which can be seen as a simplified version of the protocol seen in Example~\ref{exa:brands-and-chaum}, but does not take into account the coordinates of the messages. Two cryptographic primitives are combined: exclusive-or over a set of nonces and a commitment scheme. Exclusive-or is defined with the following properties:

{\scriptsize
\begin{verbatim}
sort NNSet .
subsorts Nonce Secret < NNSet .
  
op null : -> NNSet .
op _*_ : NNSet NNSet -> NNSet [assoc comm] .
vars X Y : [NNSet] .

eq [idem] :     X * X = null    [variant] .
eq [idem-Coh] : X * X * Y = Y   [variant] .
eq [id] :       X * null = X    [variant] .
\end{verbatim}
}

\noindent
The commitment scheme allows a participant to commit to a chosen hidden value at an early protocol stage and reveal it later. It is defined with the following properties:

{\scriptsize
\begin{verbatim}
op commit : Nonce Secret -> NTMsg .
op open : Nonce Secret NTMsg -> [Boolean] .
eq open(N1:Nonce,Sr:Secret,commit(N1:Nonce,Sr:Secret)) = true [variant] . 
\end{verbatim}
}

\noindent
The \texttt{open} function is defined only for the successful case. This implies the use of the kind \texttt{[Boolean]} rather than the sort \texttt{Boolean}. We also use additional operators for this protocol, which allow us to define signing, message concatenation, and the creation of nonces and secrets. 

{\scriptsize
\begin{verbatim}
sorts Msg NTMsg TMsg .
sorts Name Honest Intruder Fresh Secret Nonce .
subsorts NNSet < NTMsg < Msg .
subsorts Nonce Secret < NNSet .
subsort Name < Msg .
subsort Honest Intruder < Name .

ops a b : -> Honest .
op i : -> Intruder .
ops ra1 rb1 rb2 : -> Fresh .
op n : Name Fresh -> Nonce .
op s : Name Fresh -> Secret .
op sign : Name NTMsg -> NTMsg .
op _;_ : NTMsg  NTMsg  -> NTMsg [gather (e E)] .
\end{verbatim}
}

\noindent
Additionally, we add several operators that will allow us to add metadata to the messages. In them, the sending and receiving times of the messages will be saved, as well as the identifier of the sender and the receiver. 

{\scriptsize
\begin{verbatim}
sorts TimeInfo NameTime NameTimeSet .
subsort NameTime < NameTimeSet .
subsort TMsg < Msg .

op _@_ : NTMsg TimeInfo -> TMsg .
op _:_ : Name Real -> NameTime .
op mt : -> NameTimeSet .
op _#_ : NameTimeSet NameTimeSet -> NameTimeSet [assoc comm id: mt] .
op _->_ : NameTime NameTimeSet -> TimeInfo .
\end{verbatim}
}

\noindent
Note that times will be represented as real numbers, one of the data types handled by Maude's SMT interface.
The distance between two participants $A$ and $B$, informally written in the paper as $d(A,B)$ is represented by an SMT variable \texttt{dab:Real}; this 
follows the approach 
of
\cite{DBLP:conf/birthday/Aparicio-Sanchez21}
and
is a generalization of the presentation of the protocol 
in \cite{DBLP:conf/indocrypt/Aparicio-Sanchez20}.

The module defined with the previous sorts, operators and rules allows us to code the strands of the Brands and Chaum protocol of Example~\ref{exa:brands-and-chaum} only with time. This will be done in the call to the narrowing algorithm, with an initial state and a target state. In the initial state, the strand sets will contain a list of messages and constraints to be processed and a list of messages and constraints processed, which will be empty. In the target state, the lists will have been inverted, so that all the messages and restrictions to be processed become processed. Consider, for example, the strands of a prover and a verifier in a regular execution of the Brands and Chaum protocol with time. With our syntax, they would be specified in the initial state as follows:

{\scriptsize
\begin{verbatim}
--- Alice, verifier
([nilEe,
  -(Commit:NTMsg                                @ b : t1:Real -> a : t1':Real),
        {(t1':Real === t1:Real + dab:Real) and dab:Real > 0/1},
  +(n(a,ra1)                                    @ a : t2:Real -> b : t2':Real),
  -((n(a,ra1) * NB:Nonce)                       @ b : t3:Real -> a : t3':Real),
        {(t3':Real === t3:Real + dab:Real) and dab:Real > 0/1 and t3:Real >= t2':Real},
        {(t3':Real - t2':Real) <= (2/1 * dab:Real) and dab:Real > 0/1},
  -(SB:Secret                                   @ b : t4:Real -> a : t4':Real),
        {open(NB:Nonce,SB:Secret,Commit:NTMsg)},
        {(t4':Real === t4:Real + dab:Real) and dab:Real > 0/1 and t4:Real >= t3':Real},
  -(sign(b,(n(a,ra1) * NB:Nonce) ; n(a,ra1))    @ b : t5:Real -> a : t5':Real),
        {(t5':Real === t5:Real + dab:Real) and dab:Real > 0/1 and t5:Real >= t4':Real}
| nileE] 
& 
--- Bob, prover
[nilEe,
  +(commit(n(b,rb1),s(b,rb2))                   @ b : t1:Real -> a : t1':Real),
  -(NA:Nonce                                    @ a : t2:Real -> b : t2':Real),
        {(t2':Real === t2:Real + dab:Real) and dab:Real > 0/1 and t2:Real >= t1':Real},
  +((NA:Nonce * n(b,rb1))                       @ b : t3:Real -> a : t3':Real),
  +(s(b,rb2)                                    @ b : t4:Real -> a : t4':Real),
  +(sign(b,(NA:Nonce * n(b,rb1)) ; NA:Nonce)    @ b : t5:Real -> a : t5':Real) 
| nileE])
\end{verbatim}
}

\noindent
We can see how the prover, Bob, will first send a commit to the verifier. Afterwards, the verifier, Alice, sends her nonce to the prover. Subsequently, the prover will send the exclusive-or of his nonce with the received one, and then sends the secret. The verifier will open it to confirm everything is okay. Finally, the prover will send the signed messages. An \texttt{@} operator appears in each message, after which the sending and receiving times of the message are saved, as well as the identifier of the sender and receiver. We can also see how SMT constraints are introduced after each received message. In them, conditions to be met are specified regarding the delivery and reception times. Conditions to satisfy relative to distances are also specified. For example, in the SMT constraint that is introduced on the strands of the prover, it is specified that the arrival time of the received message must be equal to its departure time plus the distance between the prover and the verifier. It is also specified that this distance must be greater than zero, and that the sending time of the message must be equal to or greater than the time in which the previous message was received.

Using this syntax and coding methodology, we have defined four experiments in which we test a regular execution of the protocol, a mafia-like attack pattern (both unconstrained and constrained), and a hijacking-like attack pattern.
In a regular execution, we get a solution, which is expected, since if the protocol is well defined, this execution should be possible. In the case of the mafia attack, a priori, a solution is also found, which translates into a possible vulnerability. However, after adding the triangle inequality $(d(a,i) + d(b,i)) > d$ as the initial constraint together with the constraint $d(V,P) > d > 0$, no solution is found. This is because, as part of the protocol definition (see Example \ref{exa:brands-and-chaum}), in this execution it is necessary that $2 \ast d(V, I) + 2 \ast d(P, I) \leq 2 \ast d$. As mentioned in Section~\ref{sec:implementation}, the initial SMT constraints can be written in one of the arguments of the narrowing command. However, it is possible to perform a hijacking attack, and that is why by specifying this pattern in one of the experiments, a solution is found. The attack occurs when an intruder located outside the neighborhood of the verifier (i.e., $d(V, I) > d$) succeeds in convincing the verifier that he is inside the neighborhood by exploiting the presence of an honest prover in the neighborhood (i.e., $d(V,P) \leq d$).

Furthermore, we have performed the two experiments using both standard narrowing with SMT constraints and canonical narrowing with SMT constraints. In this way, we can make a performance comparison when it comes to finding (or not) the mentioned attacks and executions. Table \ref{tab:brands-and-chaum-time} shows the results obtained. We do not specify the indicated depth limit for the search tree since we have not limited that parameter, just because we now that the search tree is finite in all these cases.

\begin{table}[ht]
    \centering
    \setlength\doublerulesep{0.5pt}
    \caption{Experiments using the Brands and Chaum with time module and standard/canonical narrowing with SMT constraints.}
    \label{tab:brands-and-chaum-time}
    \begin{tabular}{ |c|c|c|c| } 
        \hline
        Reachability problem & Algorithm & Execution time & Solutions found \\
        \hhline{====}
        Regular execution             & standard          & 900 ms      & 1    \\
        \hline
        Regular execution             & canonical         & 960 ms      & 1    \\
        \hline
        Unconstrained mafia attack    & standard        & 10108 ms     & 4    \\
        \hline
        Unconstrained mafia attack    & canonical       & 8604 ms     & 4    \\
        \hline
        Constrained mafia attack      & standard        & 8828 ms     & 0    \\
        \hline
        Constrained mafia attack      & canonical       & 5644 ms     & 0    \\
        \hline
        Hijacking attack              & standard        & 47084 ms    & 8    \\
        \hline
        Hijacking attack              & canonical       & 10332 ms    & 8    \\
        \hline
    \end{tabular}
\end{table}

In the results of the experiments, the performance improvement of the canonical narrowing with respect to the standard narrowing can be clearly observed, once again. Although the number of solutions is the same using both algorithms, in general, a significant improvement can be seen in the execution times. For the regular execution, which is relatively simple, we get very similar execution times, since the search space is not large enough for canonical narrowing to have an impact. However, when testing executions of attacks, which are represented with more complex reachability problems, canonical narrowing times improve up to 5 times those of the standard narrowing. The improvement is given both for the hijacking attack, which returns the mentioned result, as well as for the mafia attack in its two versions. The first (unconstrained) represents an unrealistic attack in which the initial restrictions that occur in the mafia attack are not taken into account. The second does take these restrictions into account, which are easily specified in a parameter of our command.

%%%%%%%%%%%%%%%%%%%%%%%%%%%%%%%%%%%%%%%%%%%%%%%%%%%%%%%%%%%%%%%%%%%%%%%%%%%%%%%%%%%%%%%%
\subsection{Brands and Chaum with time and space}
\label{subsec:brands-chaum-time-space}
An extension of the above protocol in which space is taken into account in addition to time is possible: Brands and Chaum with time and space, detailed at a theoretical level in Example~\ref{exa:brands-and-chaum-2}. In this case, the coordinates related to the sending and receiving of each message appear in the metadata of the messages and in the restrictions, that is, the coordinates of the participants. To be able to write this, a slight modification of the previous protocol specification (see Section~\ref{subsec:brands-chaum-time}) is enough, as well as the addition of a new operator:

{\scriptsize
\begin{verbatim}
sort CoordNameTime .
op _:_,_,_,_ : Name Real Real Real Real -> CoordNameTime .
op _->_ : CoordNameTime NameTimeSet -> TimeInfo .
\end{verbatim}
}

\noindent
In this section,
the distance between two participants $A$ and $B$, informally written in the paper as $d(A,B)$ is also represented by an SMT variable \texttt{dab:Real}.

Once the modification is done, it is possible to encode the new strands. For example, the strands for a verifier and a prover in a regular execution of the protocol would now be as follows:

{\scriptsize
\begin{verbatim}
--- Alice, verifier
[nilEe, 
  -(Commit:NTMsg                    
        @ b : x1:Real,y1:Real,z1:Real,t1:Real -> a : t2:Real), 
            {(t2:Real === t1:Real + dab1:Real) and (dab1:Real > 0/1) and
            ((dab1:Real * dab1:Real) === (((x1:Real - ax:Real) * (x1:Real - ax:Real)) + 
            ((y1:Real - ay:Real) * (y1:Real - ay:Real))) + 
            ((z1:Real - az:Real) * (z1:Real - az:Real)))},      
  +(n(a,ra1)
        @ a : ax:Real,ay:Real,az:Real,t2:Real -> b : t3:Real),
  -((n(a,ra1) * NB:Nonce)
        @ b : x3:Real,y3:Real,z3:Real,t3:Real -> a : t4:Real), 
            {(t4:Real === t3:Real + dab3:Real) and (dab3:Real > 0/1) and
            ((dab3:Real * dab3:Real) === (((x3:Real - ax:Real) * (x3:Real - ax:Real)) + 
            ((y3:Real - ay:Real) * (y3:Real - ay:Real))) + 
            ((z3:Real - az:Real) * (z3:Real - az:Real)))},      
            {((t4:Real - t2:Real) <= (2/1 * d:Real)) and (d:Real > 0/1)},
  -(SB:Secret                                   
        @ b : x4:Real,y4:Real,z4:Real,t5:Real -> a : t6:Real),
            {open(NB:Nonce,SB:Secret,Commit:NTMsg)},
            {(t6:Real === t5:Real + dab4:Real) and (dab4:Real > 0/1) and 
            ((dab4:Real * dab4:Real) === (((x4:Real - ax:Real) * (x4:Real - ax:Real)) + 
            ((y4:Real - ay:Real) * (y4:Real - ay:Real))) + 
            ((z4:Real - az:Real) * (z4:Real - az:Real)))},  
  -(sign(b,(n(a,ra1) * NB:Nonce) ; n(a,ra1))    
        @ b : x5:Real,y5:Real,z5:Real,t7:Real -> a : t8:Real), 
            {(t8:Real ===  t7:Real + dab5:Real) and (dab5:Real > 0/1) and
            ((dab5:Real * dab5:Real) === (((x5:Real - ax:Real) * (x5:Real - ax:Real)) + 
            ((y5:Real - ay:Real) * (y5:Real - ay:Real))) + 
            ((z5:Real - az:Real) * (z5:Real - az:Real)))} 
| nileE]
&
--- Bob, prover
[nilEe,
  +(commit(n(b,rb1),s(b,rb2))                           
        @ b : bx:Real,by:Real,bz:Real,t1:Real -> a : t2:Real),
  -(NA:Nonce                                            
        @ a : x2:Real,y2:Real,z2:Real,t2:Real -> b : t3:Real),
            {(t3:Real === t2:Real + dab2:Real) and (dab2:Real > 0/1) and
            ((dab2:Real * dab2:Real) === (((x2:Real - bx:Real) * (x2:Real - bx:Real)) + 
            ((y2:Real - by:Real) * (y2:Real - by:Real))) + 
            ((z2:Real - bz:Real) * (z2:Real - bz:Real)))},
  +((NA:Nonce * n(b,rb1))                               
        @ b : bx:Real,by:Real,bz:Real,t3:Real -> a : t4:Real),
  +(s(b,rb2)                                            
        @ b : bx:Real,by:Real,bz:Real,t3:Real -> a : t6:Real),
  +(sign(b,(NA:Nonce * n(b,rb1)) ; NA:Nonce)            
        @ b : bx:Real,by:Real,bz:Real,t3:Real -> a : t8:Real)                        
| nileE]
\end{verbatim}
}

\noindent
The exchange of messages is very similar to what we have seen before, but in this case the metadata is somewhat more complex, since the sending coordinates are attached to each sending time. In addition, the restrictions are more complex, since in this case it will also be necessary to verify that the conditions required for those coordinates are satisfied at each moment. In fact, since the new constraints are non-linear arithmetic, Maude's SMT is not capable of processing them. In order to correctly execute the traces related to this protocol, we have used a version of Maude called Maude-NRA, which provides an SMT solver (e.g., Yices2~\cite{10.1007/978-3-319-08867-9_49}) that is capable of processing this type of arithmetic constraints.

We have performed experiments for this protocol with a regular execution, a mafia-like attack pattern, and a hijacking-like attack pattern. The results are similar to the previous ones, although more complex. Regular execution returns a solution, since it is possible to do it without problems. The hijacking attack is again possible as well, so a solution is again returned. Regarding the mafia attack, the same thing happens: a priori it is possible, but by adding the initial SMT constraints necessary for the trace to be consistent, the attack is impossible (see Example~\ref{exa:brands-and-chaum}). These restrictions are the same as before, but in this case some relative to coordinates are also added. In Table \ref{tab:brands-and-chaum-time-space} we show the results of these runs using the standard narrowing and canonical narrowing algorithms.

\begin{table}[ht]
    \centering
    \setlength\doublerulesep{0.5pt}
    \caption{Experiments using the Brands and Chaum with time and space module and standard/canonical narrowing with SMT constraints.}
    \label{tab:brands-and-chaum-time-space}
    \begin{tabular}{ |c|c|c|c| } 
        \hline
        Reachability problem & Algorithm & Execution time & Solutions found \\
        \hhline{====}
        Regular execution               & standard        & 1624 ms     & 1    \\
        \hline
        Regular execution               & canonical       & 2428 ms     & 1    \\
        \hline
        Unconstrained mafia attack    & standard        & 17044 ms    & 4    \\
        \hline
        Unconstrained mafia attack    & canonical       & 16364 ms    & 4    \\
        \hline
        Constrained mafia attack      & standard        & 15412 ms    & 0    \\
        \hline
        Constrained mafia attack      & canonical       & 11100 ms    & 0    \\
        \hline
        Hijacking attack              & standard        & 42276 ms    & 8    \\
        \hline
        Hijacking attack              & canonical       & 11760 ms    & 8    \\
        \hline
    \end{tabular}
\end{table}

Note that in the extended version with space, the complexity increases, so the execution times increase with respect to Table \ref{tab:brands-and-chaum-time}. As in Brands and Chaum with time (see Section \ref{subsec:brands-chaum-time}), disregarding space, regular execution yields very similar runtime results for both algorithms. This is because its specification is done with a very simple reachability problem. But if we pay attention to the rest of the experiments, in which the reachability problems are more complex, we again see a significant improvement of the canonical narrowing with respect to the standard narrowing. In this case, the execution times are higher for both algorithms, since by taking into account the space, the complexity increases when generating the search space and computing the solutions. However, we see an improvement in execution time of almost 4 times when using the canonical narrowing algorithm.

Both with this example and with the previous one, we check how our algorithm is capable of handling SMT constraints in a correct and efficient way, allowing us to specify protocols that contain conditional rules. At the same time, we show the positive impact that the use of canonical narrowing has in the world of cryptographic protocols.

%%%%%%%%%%%%%%%%%%%%%%%%%%%%%%%%%%%%%%%%%%%%%%%%%%%%%%%%%%%%%%%%%%%%%%%%%%%%%%%%%%%%%%%%
%%%%%%%%%%%%%%%%%%%%%%%%%%%%%%%%%%%%%%%%%%%%%%%%%%%%%%%%%%%%%%%%%%%%%%%%%%%%%%%%%%%%%%%%
\section{Conclusions and future work}
\label{sec:conclusions}
We have presented the extended version of two articles already published \cite{DBLP:conf/wrla/Lopez-Rueda22,DBLP:conf/wrla/Lopez-Rueda22-invited}. We have unified the work of both in an improved algorithm, using new examples and experiments that demonstrate its potential. Rewrite theories of different complexities have been used, evidencing the correct operation of narrowing with irreducibility constraints and SMT. In addition, new experiments have been defined, which together with the previous ones (which have been re-executed), demonstrate the superiority of canonical narrowing in terms of performance compared to standard narrowing in most cases.

The contributions of this work can be very useful in all those protocol analysis tools that rely on Maude unification or narrowing, such as Maude-NPA. In addition, they can be fitted in other areas that can use narrowing to perform symbolic analysis. Some of them might be logical model checking verification of concurrent systems, theorem proving or partial evaluation. 

As future work, we intend to integrate into the algorithm the ability to handle conditional rewrite theories at a more generic level. Likewise, a graphical interface can be assembled that brings together all the possible calls to the algorithm, making it easier for the user to use and read the results. Finally, it would be interesting to 
combine canonical narrowing with SMT constraints with the computation of most general variant unifiers \cite{DBLP:journals/corr/abs-1909-08241} when performing unification. The ultimate goal is to achieve a unified algorithm that allows users to configure different narrowing variants to use the one that best suits each case, with optimal performance and clear results.

%%%%%%%%%%%%%%%%%%%%%%%%%%%%%%%%%%%%%%%%%%%%%%%%%%%%%%%%%%%%%%%%%%%%%%%%%%%%%%%%%%%%%%%%
%%%%%%%%%%%%%%%%%%%%%%%%%%%%%%%%%%%%%%%%%%%%%%%%%%%%%%%%%%%%%%%%%%%%%%%%%%%%%%%%%%%%%%%%

%\newpage

\bibliographystyle{elsarticle-num}
\bibliography{ref.bib}

\end{document}